\documentclass[journal]{IEEEtran}

\usepackage{fancyhdr}
\usepackage{hyperref}

\usepackage{amssymb}
\usepackage{amsmath}
\usepackage{amsfonts}
\usepackage{subfig}
\usepackage{bbm}
\usepackage[numbers]{natbib}
\usepackage{graphicx}
\usepackage{epstopdf}
\usepackage{color}
\usepackage{amsthm}
\usepackage{caption}

\usepackage{algorithm,algorithmic}

\newtheorem{thm}{Theorem}[section]

\newtheorem{lemma1}{Theorem}[section]
\newtheorem{lemma}[lemma1]{Lemma}

\begin{document}

\title{Using Node Centrality and Optimal Control to Maximize Information Diffusion in Social Networks}

\author{Kundan Kandhway and Joy Kuri% <-this % stops a space
\thanks{Authors are with the Department of Electronic Systems Engineering, Indian Institute of Science, Bangalore, 560012, India  (e-mail: \{kundan, kuri\}@dese.iisc.ernet.in) .}% <-this % stops a space
%\thanks{Manuscript received April 19, 2005; revised December 27, 2012.}
}

% make the title area
\maketitle
%\IEEEpeerreviewmaketitle

\thispagestyle{fancy}

\begin{abstract}
We model information dissemination as a susceptible-infected epidemic process and formulate a problem to jointly optimize seeds for the epidemic and time varying resource allocation over the period of a fixed duration campaign running on a social network with a given adjacency matrix. Individuals in the network are grouped according to their centrality measure and each group is influenced by an external control function---implemented through advertisements---during the campaign duration. The aim is to maximize an objective function which is a linear combination of the reward due to the fraction of informed individuals at the deadline, and the aggregated cost of applying controls (advertising) over the campaign duration. We also study a problem variant with a fixed budget constraint. We set up the optimality system using Pontryagin's Maximum Principle from optimal control theory and solve it numerically using the forward-backward sweep technique. Our formulation allows us to compare the performance of various centrality measures (pagerank, degree, closeness and betweenness) in maximizing the spread of a message in the optimal control framework. We find that degree---a simple and local measure---performs well on the three social networks used to demonstrate results: scientific collaboration, Slashdot and Facebook. The optimal strategy targets central nodes when the resource is scarce, but non-central nodes are targeted when the resource is in abundance. Our framework is general and can be used in similar studies for other disease or information spread models---that can be modeled using a system of ordinary differential equations---for a network with a known adjacency matrix.
\end{abstract}


\IEEEpeerreviewmaketitle

\section{Introduction}

\IEEEPARstart{A}{ll} major brands and even small businesses have presence in multiple online social media these days. They use social media to promote their products and services. Movies have used social networks for creating strong buzz before their release (\emph{e.g.}, the case of `Hunger Games' \cite{taepke2012hunger}). Paramount Pictures---a movie distributor---and Twitter---a social media company---successfully partnered to promote the movie `Super 8' \cite{shortyawards2011super8}. The success of some of charity campaigns, like Movember (which raised over $\pounds 63.9$m in 2013 for men's health) and \#nomakeupselfie (which helped Cancer Research raise about $\pounds 8$m in a week) would have been impossible without the attention they got in the social media \cite{zoe2014five}. These examples establish the central role played by social networks in information dissemination.

In this paper, we devise optimal strategies to maximize the spread of a message in a social network for a fixed duration campaign. Such problems lie in the broad class of \emph{influence maximization problems}, introduced first in \cite{domingos2001mining}. In \cite{domingos2001mining}, the authors found a set of initial seeds such that their final influence on the network is maximum. This problem was shown to be NP-hard for many spread models, such as linear threshold and independent cascade in \cite{kemp2003maximizing}, and performance of the greedy optimization algorithm was analyzed. The works above averaged multiple runs of Monte-Carlo simulations to estimate influence spread, making the solution methodology computationally inefficient. In \cite{kimura2007extracting, saito2012efficient, ohara2013predictive}, the authors used tools from statistical physics (bond percolation) and predictive modeling to approximately solve the greedy optimization and seed selection problems. In \cite{chen2009efficient, chen2010scalable}, computationally efficient heuristics were introduced to solve the same problem that scaled to large network sizes. Monotonicity and submodularity of the spread functions were used to further speed up the greedy optimization in \cite{goyal2011celf}.

All of the above benchmark studies on influence maximization focused on optimizing the seed nodes---the individuals who start the spreading process at time $t=0$. Once the seeds are selected, the decision maker does not influence the system during the rest of the campaign horizon, $t>0$. In contrast, this paper formulates the resource allocation problem throughout the campaign horizon.

Our approach is as follows. Information diffusion is modeled as a susceptible-infected (SI) epidemic, which is represented mathematically as a set of differential equations. The population is divided into two classes---susceptible individuals are unaware of the message, while the infected ones are informed and spreading the message. We group the nodes of the network, and each group is influenced by an advertisement strategy (control function) throughout the campaign duration which transfers susceptible individuals to the infected class; this is in addition to the epidemic message transfer in the system. The strategies are computed such that a net reward function---a linear combination of the reward due to information spread and the advertisement costs---is maximized.

Our main contribution in this work is to formulate the information diffusion maximization problem on a network with known adjacency matrix as a \emph{joint seed--optimal control problem where nodes can be targeted based on their centrality measures in the network. Here the campaigner has the flexibility to allocate the resource throughout the campaign horizon.} This is done by adjusting a control function---\emph{e.g.}, rate at which advertisements are shown to individuals in their social network timeline during the campaign horizon $[0,T]$---such that the reward function is maximized. Unlike traditional optimization problems, \emph{optimal control problems optimize functions and not real variables} to maximize an objective. These controls affect the evolution of the system which is captured by a set of ordinary differential equations. In addition, our framework jointly optimizes the seeds for the epidemic. Details on optimal control theory and solution techniques using Pontryagin's Maximum Principle can be found in \cite{kirk2012optimal}.

In practice, the advertisements are shown in the social network timeline of the individuals throughout the campaign horizon and not just once at the beginning of the campaign. Thus, the joint seed--optimal control framework developed in this paper is more suitable for real world scenarios than only seed optimization which has been carried out extensively in the literature. There are many companies/entities interested in advertising products/running campaigns in the network. The social network economic planner would want to entertain as many of them as possible at any time instant. Thus, advertisements by a single entity are spread over time in real world situations. \emph{This motivates the formulation of the influence maximization problem as an optimal control problem rather than just a seed optimization problem.}

\emph{Justification for Using SI as the Spread Model:} The similarity between biological epidemics and information dissemination in a network motivates us to use the SI model. In the uncontrolled system, the message is passed to susceptible neighbors of an infected/informed node at some rate due to posts made by the infected/informed node in its social network timeline, similar to the spread of a disease.

The SI model is suitable for information epidemics when spreaders do not recover (stop spreading) during the duration of the campaign. This is possible for campaigns with small deadlines (\emph{e.g.}, charity campaigns, movie promotion after its release). Advertising individual matches in a long tournament (\emph{e.g.}, Football World Cup) where deadlines are short because gaps between the matches are small, and interest of the population is high, is another example of this situation. If the campaign generates lot of interest in the population, the SI process is again a suitable model for information diffusion (\emph{e.g.}, election/political campaigns).

Prior works have used the SI process to model information dissemination in a population. For example, information dissemination in a social call graph---where two individuals are connected if they spoke over the mobile network---was modeled as an SI process in \cite{onnela2007structure}. Information dissemination in social and technological networks was modeled as an SI process in \cite{khelil2002epidemic,busch2012homogeneous}. Information diffusion due to face-to-face interactions was modeled as an SI process in \cite{isella2011s}.

We emphasize that the \emph{framework developed in this paper is not limited to the SI model} and can be easily extended to other ordinary differential equation based epidemic models on networks with given adjacency matrix, such as, susceptible-infected-susceptible/recovered (SIS/R) and Maki-Thompson rumor models.

\subsection{Related Work and Our Contributions}

We have already discussed the benchmark studies on influence maximization that address the optimal seed selection problem. More recent works have generalized the seed selection algorithm for more general networks and/or competing opinions. For example, networks with signed edges (friend/foe) were considered in \cite{li2013influence}, temporal networks were considered in \cite{michalski2014seed} and competing opinions (positive/negative) were considered in \cite{liu2010influence, zhang2013maximizing}.

All of the above works identified good seeds to maximize informed/influenced individuals in the network. After the seed was selected, the system evolved in an uncontrolled manner. However, as discussed earlier, in real world situations we see advertisements and promos appearing throughout the campaign duration. \emph{Our formulation not only identifies good seeds, but also shapes system evolution during the campaign horizon through controls.} It is more general and more practical with respect to the above works.

Previous works such as \cite{kandhway2014run, karnik2012optimal, kandhway2014optimal} maximized information epidemics under different diffusion models by formulating optimal control problems; \emph{however, all of them assumed no underlying network structure.} Homogeneous mixing---in which any two individuals in the system are equally likely to meet and interact for the purpose of information diffusion---is not a plausible assumption for individuals connected via social networks. In addition, \emph{seed optimization was not carried out} because homogeneous mixing was assumed.

Works such as \cite{asano2008optimal, zhu2012optimal, khouzani2011optimal} \emph{discussed prevention of biological and computer virus epidemics} using optimal control strategies. Again \emph{homogeneous mixing was assumed.} \cite{youssef2013mitigation} formulated a biological epidemic mitigation problem for a networked population, but presented optimal results for only a five node network (heuristics were proposed for larger but synthetic networks). \emph{The effect of different centrality measures was not studied} in that work. Also, we present and interpret results for \emph{real social networks}. Similarly, \cite{dayama2012optimal} devised product marketing strategies for a synthetic network with only two types of individuals---ones with high degrees and low degrees, which is only a \emph{crude approximation of real social networks.}

In the following we list the \emph{\textbf{main contributions of this paper}}. \emph{\textbf{First}}, we formulate a joint seed--optimal control problem to allocate seeds and time varying resource to groups of nodes in the network. This is a dynamic optimization framework rather than a static one. The objective function is a linear combination of the final fraction of informed nodes and the aggregate cost of application of controls. We group nodes based on their centrality values in the network, which allows us to compare the performance of various centrality measures in maximizing information spread in the framework of optimal control. However, we emphasize that grouping can be carried out using other methods as well. We also study a problem variant which has a fixed budget constraint (with given seed). \emph{\textbf{Second}}, we prove some simple structural results for the optimal control signal analytically. \emph{\textbf{Third}}, we present numerical schemes that use the forward-backward sweep technique (generalized for a network setting) to solve the optimality systems. We modify this scheme to solve the problem with a budget constraint and for joint optimization. These schemes are capable of handling large optimality systems and allow us to present and interpret results for real social networks---scientific collaboration, Slashdot and Facebook. We also study the effect of system parameters on the optimal strategy and the objective function. These schemes can be easily adapted to study other differential equations based spread models on networks with given adjacency matrix (\emph{e.g.}, SIS/R, SIRS, Maki-Thompson etc.).

We believe that our framework and formulation provide novel insights into the problem of maximizing influence through optimal controls. It allows us to \emph{\textbf{answer the following questions that were overlooked in the previous literature:}} (1) Depending on scarcity/abundance of resource, whom to target---central nodes who are better spreaders, or disadvantaged non-central nodes who would not receive the message through epidemic spreading? How is this resource allocated over time? (2) How does the simple and local centrality measure---node degree---perform compared to more complicated ones like betweenness, closeness  or pagerank? (3) How much advantage do optimal strategies---which heterogeneously allocate resource over time and node groups---achieve over simple heuristic strategies which target everyone equally at the same rate throughout the campaign duration on real social networks?

\section{System Model and Problem Formulation}
\label{sec:sys_model_prob_formaulation}

We consider an undirected and non-weighted static network, $(\mathbb N,\boldsymbol A)$. The set of all nodes (vertices) in the network (graph) is represented by $\mathbb N=\{ 1,2,...,N \}$. The adjacency matrix of the network is given by $\boldsymbol A=\{A_{ij}\}$. Thus, $A_{ij}=1$ if nodes $i$ and $j$ are connected and $0$ otherwise. Also, in this work we consider an undirected network, \emph{i.e.}, $A_{ij}=A_{ji}$. Let the probability of node $j\in \mathbb N$ being in infected (informed) state at time $t$ be $i_j(t)$. Then, $i_j(t)=1-s_j(t)$, where $s_j(t)$ is the probability of the node being in susceptible state at time $t$.

We briefly describe the uncontrolled SI epidemic on a network with known adjacency matrix \cite[Sec. 17.10]{newman2009networks}, and then adapt it to include controls to accelerate information spreading. If the information spreading rate is $\beta$, the message transfers from a single infected neighbor $k$ to a susceptible node $j$, with probability $\beta dt$ in a small interval $dt$. It is required that $j$ be susceptible to begin with, which happens with probability $s_j(t)$, and the neighbor be infected, which happens with probability $A_{jk}i_k(t)$ (node $k$ is a neighbor only if there is a link between $j$ and $k$, \emph{i.e.}, $A_{jk}=1$). Aggregating over all neighbors, the evolution of the probability $i_j(t)$ in an SI epidemic in a network with adjacency matrix $\boldsymbol A$ is given by \cite[Sec. 17.10]{newman2009networks}:
\begin{align}
\dot i_j(t) & =  \beta s_j(t) \sum_{k=1}^N (A_{jk}i_k(t));~i_j(0)=x_{0j};~~\forall j\in \mathbb N, \label{eq:uncontrolled_SI}
\end{align}
where, $x_{0j}$ is the probability that node $j$ is selected as seed.

\emph{Controlled System:} The shape of the control function decides resource allocation over time $[0,T]$, identifying important and non-important periods of the campaign. In addition, to differentiate allocation of resources among nodes of different `types', we aggregate them into groups. One control influences each group. The whole network is divided into $M$ non-overlapping groups, $\mathbb N_m,~1\leq m\leq M$; that is, $\mathbb N=\cup_{m=1}^M \mathbb N_m$ and $\mathbb N_p\cap\mathbb N_q=\emptyset$ for $p\neq q$. Each node belongs to one and only one of the groups, $\mathbb N_m,~1\leq m\leq M$, and the sum of the number of nodes in all the groups, $\sum_{m=1}^M|\mathbb N_m| = N$. We use notions of node centrality listed in Sec. \ref{sec:centralities} for dividing the network into $M$ groups.

Only non-random grouping is useful, that is, nodes of similar `types' should be grouped together. In this work, we differentiate node types based on their centrality value. In the following, we illustrate how division of $N$ nodes can be carried out in $M$ groups using, for example, degree centrality. Suppose we decide to keep all the groups of the same size. We calculate the degree centralities of all nodes in the network, and then place the bottom $N/M$ nodes with least values of the centrality in group $1$, $\mathbb N_1$; next $N/M$ nodes in groups $2$, $\mathbb N_2$; and so on. If the centrality values of two or more nodes are tied and there is insufficient space in group $\mathbb N_i$, a number of nodes from the tied set are randomly chosen and placed in group $\mathbb N_i$, and the rest are placed in group $\mathbb N_{i+1}$.

We emphasize that the above method is just an example; the sizes of groups need not be equal in our formulation. Also, metrics other than centrality measures may be used for the division of the groups (\emph{e.g.}, community structure in the network). However, to present results in this paper (Sec. \ref{sec:results}), we have followed the method mentioned above to group the nodes in the network.

Denote by $p_m$ the fraction of nodes in group $m$, \emph{i.e.}, $p_m\triangleq |\mathbb N_m|/N$. The control function $u_m(t)$ aids information diffusion in the group $\mathbb N_m$ at time $t$. It represents the rate at which advertisements are placed in the social network timeline of the nodes in group $\mathbb N_m$.

We state the joint seed--control problem in the following and discuss the formulation:
\begin{subequations}
\label{eq:opt_prob}
\begin{align}
\underset{ \begin{smallmatrix} (u_1(t),...,u_M(t))^T,\\ \big\{(i_{01},...,i_{0M})^T: \\
0\leq i_{0m}\leq 1,~\forall m;\\ \sum_{m=1}^Mp_mi_{0m}=i_0\big\} \end{smallmatrix} }{\text{max}} J & = \frac{1}{N}\sum_{j=1}^N i_j(T) - \sum_{m=1}^M \int_0^Tg_m(u_m(t))dt, \label{eq:cost_funtion}\\
\text{s.t.:~~} \dot i_j(t) & =  \beta s_j(t) \sum_{k=1}^N (A_{jk}i_k(t)) + \nonumber \\
& \sum_{m=1}^M ( \mathbbm 1_{\{j\in \mathbb N_m\}} u_m(t)) s_j(t) ;~ 1 \leq j \leq N, \label{eq:opt_prob_states} \\
i_j(0) & =  \sum_{m=1}^M \mathbbm 1_{\{j\in \mathbb N_m\}} i_{0m}; ~~1 \leq j \leq N. \label{eq:opt_prob_init_cond} \end{align}
\end{subequations}

We have considered the net reward function (\ref{eq:cost_funtion}) to be a weighted combination of the reward due to extent of information dissemination (captured by $\frac{1}{N}\sum_{j=1}^N i_j(T)$) and the aggregate cost due to expenditure caused by application of controls over the decision horizon (captured by $\sum_{m=1}^M \int_0^Tg_m(u_m(t))dt$). Constant weights are subsumed in $g_m(.)$'s. In a social network, it is impractical to assume that we know the states (susceptible or infected) of the nodes a-priori. So, the advertisement is shown to both susceptible and infected nodes; however, it has useful effect on only susceptible nodes. Thus, the cost of application of control is only a function of $u_m(t)$ (and not fraction of susceptible or infected nodes in class $m$).

The control transfers nodes from susceptible to infected class. A node $j\in\mathbb N_m$ is shown the advertisement at the rate $u_m(t)$ at time $t$. In other words, in an interval $dt$ at time $t$, the probability of showing the advertisement to a node $j\in\mathbb N_m$ is $u_m(t)dt$. Since we do not know the state (susceptible or infected) of the node beforehand in an online social network, and the advertisement is shown to both categories of nodes; the advertisement has useful effect (of infecting the node) only if the node is susceptible, which happens with probability $s_j(t)$. Thus, in an interval $dt$ at time $t$, the additional increment in the probability of node $j\in\mathbb N_m$ being in infected class is $u_m(t)s_j(t)dt$. This additional increment is the second term on the right hand side of Eq. (\ref{eq:opt_prob_states}). Note that $\mathbbm 1_{\{j\in \mathbb N_m\}}=1$ when node $j$ is in group $m$ and $0$ otherwise.\footnote{The uncontrolled and controlled SI models used in Eqs. (\ref{eq:uncontrolled_SI}) and (\ref{eq:opt_prob_states}) do not consider spontaneous conversion of susceptible nodes to infected state without any external control or peer influence, \emph{i.e.}, information spreads only through the campaigner's advertisements and neighbors. Models used in previous works, \emph{e.g.}, \cite{kemp2003maximizing,kimura2007extracting,chen2009efficient,goyal2011celf} (linear-threshold and/or independent cascade), \cite{saito2012efficient} (SIS), and many others make a similar assumption where nodes are influenced only due to their peers in the network. However, spontaneous conversion, due to sources external to the system, may be handled in our formulation by adding a constant or state dependent term to the RHS of Eqs. (\ref{eq:uncontrolled_SI}) and (\ref{eq:opt_prob_states}).}

In (\ref{eq:opt_prob_init_cond}), $i_{0m}$ with $0\leq i_{0m}\leq 1$, represents the fraction of nodes selected as seeds in group $\mathbb N_m$. In other words, each node is selected as seed with probability $i_{0m}$ in group $\mathbb N_m$. Thus, for a node $j$ belonging to group $\mathbb N_m$, the initial condition is as represented in Eq. (\ref{eq:opt_prob_init_cond}). Also, the constant $i_0$ in the constraint $\sum_{m=1}^Mp_mi_{0m}=i_0$ represents the `seed budget' (fraction of total population that can be selected as seed at $t=0$).

We assume the cost functions $g_m(.)$'s to be strictly convex and increasing functions in their arguments $u_m(t)\geq 0$, with $g_m(0)=0$ (stronger control leads to greater cost). Costs are convex functions in many economic applications \cite{newman1987convexity}. For our scenario where we attempt to maximize the spread of useful information, the controls $u_m(t)$'s are never negative. Instead of explicitly incorporating this constraint in Problem (\ref{eq:opt_prob}), we ensure it by assuming $g_m(.)$'s to be even functions (by taking an even extension of the cost function defined for $u_m(t)\geq 0$). If we do so, negative values of controls incur positive costs but reduce the reward function (\ref{eq:cost_funtion}) by reducing the fraction of infected nodes in the population. Instead, if $|u_m(t)|$ is applied, the reward increases and expenditure is the same. Thus, for any $t$, the control is never negative. No constraint on the control function $u_m$'s makes the theoretical analysis in Sec. \ref{sec:analysis_solution_structure} simple. Also, this assumption ensures that we get a correct solution from the numerical schemes (discussed later in Sec. \ref{sec:num_scheme}).

We have presented the joint seed and time varying resource allocation optimization framework for the SI model in this paper. However, the framework can be easily extended to other differential equation based spread models such as susceptible-infected-susceptible/recovered (SIS/R), Maki-Thompson model and other derivatives like SIRS etc. To do this, one should replace Eq. (\ref{eq:opt_prob_states}) with the differential equations corresponding to the new model.

\section{Analysis and Solution of Problem (\ref{eq:opt_prob})}
\label{sec:analysis_solution}

We first set up the optimality system by assuming the seeds are given (known) in Sec. \ref{sec:soln_by_pontryagin}, and present an algorithm to compute the controls corresponding to this system numerically in Sec. \ref{sec:num_scheme_fixed_seed}. We will use this algorithm in the numerical scheme in Sec. \ref{sec:num_scheme_joint} to jointly optimize the seed and dynamic resource allocation.

\subsection{Pontryagin's Maximum Principle}
\label{sec:soln_by_pontryagin}

We get the following optimality system for Problem (\ref{eq:opt_prob}) when seeds are given (details of Pontryagin's Principle can be found in \cite{kirk2012optimal}):

\emph{Hamiltonian:} $H(\boldsymbol i(t),\boldsymbol \lambda(t),\boldsymbol u(t))  = - \sum\limits_{m=1}^M g_m(u_m(t)) + \sum\limits_{l=1}^N\lambda_l(t)\big( \beta s_l(t) \sum_{k=1}^N (A_{lk}i_k(t)) \hspace{-.1em}+\hspace{-.4em} \sum\limits_{m=1}^M ( \mathbbm 1_{\{l\in \mathbb N_m\}} u_m(t)) s_l(t) \big).$
\\ \emph{State equations:} Eqs. (\ref{eq:opt_prob_states}) with $i_j(t)$ and $s_j(t)$ replaced by $i_j^*(t)$ and $s_j^*(t)$ for all $j\in\mathbb N$. The initial condition is given by Eq. (\ref{eq:opt_prob_init_cond}) with $i_j(0)$ replaced by $i_j^*(0)$ for all $j\in\mathbb N$.
\\ \emph{Adjoint equations:}
\begin{align}
\dot{\lambda}_j^*(t)=-\frac{\partial H}{\partial i_j(t)} = -\beta\sum_{l=1}^N\lambda_l^*(t) ( s_l^*(t)A_{lj} ) + \nonumber \\
\beta\lambda_j^*(t)\sum_{k=1}^N(A_{jk}i_k^*(t)) + \lambda_j^*(t) \sum_{m=1}^M ( \mathbbm 1_{\{j\in \mathbb N_m\}} u_m^*(t)). \label{eq:adjoint_eq}
\end{align}
\\ \emph{Hamiltonian maximizing condition:}
\begin{align}
\boldsymbol u^*(t)=\text{argmax}\{ H(\boldsymbol i^*(t),\boldsymbol \lambda^*(t),\boldsymbol u^*(t)) \} \nonumber \\
\Rightarrow \frac{\partial H}{\partial u_m(t)} = -g'_m(u_m^*(t)) + \sum_{l\in\mathbb N_m}\lambda_l^*(t) s_l^*(t) = 0 \nonumber \\
\Rightarrow g'_m(u_m^*(t)) = \sum_{l\in\mathbb N_m}\lambda_l^*(t) s_l^*(t),~1\leq m \leq M, \label{eq:hamil_max_cond_control_val} \\
\Rightarrow u_m^*(t) = g_m'^{-1}\Big(\sum_{l\in\mathbb N_m}\lambda_l^*(t) s_l^*(t)\Big),~1\leq m \leq M. \label{eq:hamil_max_cond_control_val1}
\end{align}
\\ \emph{Transversality condition:}
\begin{align}
\lambda_j^*(T)=1/N,~1\leq j\leq N. \label{eq:transversaility_cond}
\end{align}

\subsection{Structural Properties of Controls}
\label{sec:analysis_solution_structure}

For the special case of $M=N$, \emph{i.e.}, each node has an individual control, it is possible to show that the controls are non-increasing functions of time (Theorem \ref{thm:control_noninc}). Further, for quadratic cost functions, the controls are convex in nature (Theorem \ref{thm:control_convex}). We first state following lemma which will be used in subsequent theorems.

\begin{lemma}
\label{thm:lambda_nonnegative}
For $M=N$, at the optimal solution, the adjoint variables $\lambda_j^*(t)\geq 0,~1\leq j\leq M,~t\in[0,T]$.
\end{lemma}
\begin{proof}
From the Hamiltonian maximizing condition of Pontryagin's Principle,
\begin{align}
H(\boldsymbol i^*(t),\boldsymbol \lambda^*(t),u_1^*(t),...,u_j^*(t),...,u_M^*(t)) \nonumber \\
\geq H(\boldsymbol i^*(t),\boldsymbol \lambda^*(t),u_1^*(t),...,0,...,u_M^*(t)). \nonumber
\end{align}
This leads to,
\begin{align}
-g_j(u_j^*(t))+\lambda_j^*(t)s_j^*(t)u_j^*(t) \geq -g_j(0) \nonumber \\
\Rightarrow \lambda_j^*(t)s_j^*(t)u_j^*(t)\geq 0,~(\because g_j(0)=0,~g_j(u_j^*(t))\geq 0), \nonumber \\
\Rightarrow \lambda_j^*(t) \geq 0. ~(\because 0\leq s_j^*(t)\leq 1,~u_j^*(t)\geq 0). \nonumber
\end{align}
\end{proof}

\begin{thm}
\label{thm:control_noninc}
For $M=N$, the optimal controls $u_j^*(t)$ are non-increasing in $t$, for $1\leq j\leq M$.
\end{thm}
\begin{proof}
Differentiating (\ref {eq:hamil_max_cond_control_val}) with respect to $t$, for $M=N$, we get,
\begin{align}
g''_j(u_j^*(t))\dot u_j^*(t)=\dot \lambda_j^*(t)s_j^*(t)+\lambda_j^*(t)\dot s_j^*(t). \nonumber
\end{align}
After simplification, this leads to (note that $\dot s^*(t)=-\dot i^*(t)$),
\begin{align}
g''_j(u_j^*(t))\dot u_j^*(t)=-\beta s_j^*(t) \sum_{l=1}^N(\lambda_l^*(t)s_l^*(t)A_{lj}). \label{eq:gpp_udot}
\end{align}
Since $g_j(.)$'s are convex so $g''_j(u_j^*(t))\geq 0$ and $\lambda_l^*(t)\geq 0$ (Lemma \ref{thm:lambda_nonnegative}), $s_j^*(t)$ being probability $\geq 0$, so we conclude that $\dot u^*(t)\leq 0$.
\end{proof}

\begin{thm}
\label{thm:control_convex}
For $M=N$ and $g_j(u_j(t))=c_ju_j^2(t),~c_j>0$ for $1\leq j\leq M$, the optimal controls $u_j^*(t)$ are convex functions of $t$, for all $1\leq j\leq M$.
\end{thm}
\begin{proof}
Using Eq. (\ref{eq:hamil_max_cond_control_val}) (for $M=N$) in Eq. (\ref{eq:gpp_udot}),
\begin{align}
2c_j\dot u_j^*(t) = & -\beta s_j^*(t) \sum_{l=1}^N(g_l'(u_l^*(t))A_{lj}) \nonumber \\
\Rightarrow 2c_j\overset{..}{u}_j^*(t) = & -\beta s_j^*(t) \sum_{l=1}^N(g_l''(u_l^*(t))\dot u_l^*(t)A_{lj}) \nonumber \\
& - \beta \dot s_j^*(t) \sum_{l=1}^N(g_l'(u_l^*(t))A_{lj}). \nonumber
\end{align}
Since $\dot u_l^*(t)\leq 0$ (Theorem \ref{thm:control_noninc}), $g_j''(.)>0$ ($\because~g_j$'s are quadratic and $c_j>0$), $s_j^*(t)\geq 0$ (probability), $g_j'(.)>0$ (assumption of stronger control leading to greater cost), $\dot s_j^*(t)\leq 0$ (from Eq. (\ref{eq:opt_prob_states}), note that $\dot s_j^*(t)=-\dot i_j^*(t)$), so we can conclude that $\overset{..}{u}_j^*(t)\geq 0$, which proves the convexity of $u_j^*(t)$ in $t$.
\end{proof}

\subsection{Numerical Solution Using Forward-Backward Sweep}
\label{sec:num_scheme}

Sec. \ref{sec:num_scheme_fixed_seed} presents the numerical scheme to solve the optimality system in Sec. \ref{sec:soln_by_pontryagin} (when seeds are given). We use this to solve the joint seed--control problem numerically in Sec. \ref{sec:num_scheme_joint}.

\subsubsection{Fixed Seed}
\label{sec:num_scheme_fixed_seed}

Let the operator $\star$ perform element wise multiplication of two vectors of the same dimension, \emph{i.e.}, $(...,a_i,...)^T\star (...,b_i,...)^T=(...,a_ib_i,...)^T$.  Let the state, adjoint, and control variables be represented by $\boldsymbol i(t)=(i_1(t),...,i_N(t))^T=\boldsymbol 1-\boldsymbol s(t),~\boldsymbol \lambda(t)=(\lambda_1(t),...,\lambda_N(t))^T$, and $\boldsymbol u(t)=(u_1(t),...,u_M(t))^T$, and, the initial condition by $\tilde{\boldsymbol i}_0=(\tilde i_{01},...,\tilde i_{0N})^T$. Here, the $j$th element is the probability that node $j$ is selected as seed, which is a given quantity. Let $\tilde{\boldsymbol u}(t)$ be a column vector of dimension $N$ with the $j$th element being $\sum_{m=1}^M ( \mathbbm 1_{\{j\in \mathbb N_m\}} u_m(t))$. If node $j$ is in class $m$, the control $u_m$ influences it. Then, given $\tilde{\boldsymbol i}_0$, Problem (\ref{eq:opt_prob}) can be solved using the forward backward sweep method \cite{asano2008optimal} in Algorithm \ref{alg:fw_back_sweep}.

\begin{algorithm}
\small
\caption{Forward-backward sweep algorithm for Problem (\ref{eq:opt_prob}) when $\boldsymbol i_0$ is given.}
\label{alg:fw_back_sweep}
\begin{algorithmic}[1]
	\REQUIRE $u_{th}$, $T$, $\beta$, $\boldsymbol i_{0}$, $A$, $M$, $\mathbb N_m$ for $1\leq m\leq M$, $P$.
	\ENSURE The optimal control signals $u_m^*(t)$, $1\leq m\leq M$.
	\STATE Initialize: For all $t\in[0,T]$, $u_m^*(t) \leftarrow 0, 1\leq m\leq M$, and  $\lambda_j^*(t) \leftarrow 0,~\forall j\in\mathbb N$, $iter \leftarrow 0$.
	\REPEAT
		\STATE $iter\leftarrow iter+1.$
		\STATE $\boldsymbol u_{old}\leftarrow \boldsymbol u^*$.
		\STATE Use the state equations (\ref{eq:opt_prob_states}), $\dot{\boldsymbol{i}}^*(t)=\beta \boldsymbol s^*(t)\star (\boldsymbol{Ai}^*(t))$, to calculate $\boldsymbol i^*(t),~\forall t$, with initial condition $\boldsymbol i^*(0) = \tilde{\boldsymbol i}_0$. \COMMENT{Forward sweep.}
		\STATE Calculate $u_m^*,~1\leq m\leq M$ using (\ref{eq:hamil_max_cond_control_val1}).
		\STATE Use the adjoint equations (\ref{eq:adjoint_eq}), $\dot{\boldsymbol{\lambda}}^*(t)=-\beta((\boldsymbol \lambda^*(t)\star \boldsymbol s^*(t))^T \boldsymbol A)^T + \beta \boldsymbol \lambda^*(t) \star (\boldsymbol A \boldsymbol i^*(t)) + \boldsymbol \lambda^*(t) \star \tilde{\boldsymbol u}^*(t)$ to calculate $\boldsymbol \lambda^*(t)$ for all $t$, with terminal conditions $\lambda_j^*(T)=1/N,~j\in\mathbb N$. (transversaility conditions). \COMMENT{Backward sweep.}
		\STATE Calculate $u_m^*,~1\leq m\leq M$ using (\ref{eq:hamil_max_cond_control_val1}).
	\UNTIL{$|| \boldsymbol u^* - \boldsymbol u_{old} ||<u_{th}$ or $iter > P$.}
\end{algorithmic}
\end{algorithm}

Since information diffusion is captured by a system of $N$ differential equations, any \emph{optimal} strategy cannot have lesser computational complexity.  For a network of size $N$, Algorithm \ref{alg:fw_back_sweep} amounts to solving a system of $N$ differential equations $P$ times, in the worst case. Thus, the worst case computational complexity of the optimal strategy is $O(N)$.

\subsubsection{Joint Seed and Resource Allocation} 
\label{sec:num_scheme_joint}

Joint Problem (\ref{eq:opt_prob}) can be expressed as:
\begin{subequations}
\label{eq:opt_prob2}
\begin{align}
\underset{ \begin{smallmatrix}
\big\{(i_{01},...,i_{0M})^T: \\
0\leq i_{0m}\leq 1,~\forall m;\\ \sum_{m=1}^Mp_mi_{0m}=i_0\big\} \end{smallmatrix} }{\text{~~max~~}} ~~(\ref{eq:cost_funtion})~~~\text{subject to:~~} (\ref{eq:opt_prob_states}), (\ref{eq:opt_prob_init_cond}); (\ref{eq:adjoint_eq}), (\ref{eq:hamil_max_cond_control_val1}), (\ref{eq:transversaility_cond}). \nonumber
\end{align}
\end{subequations}
The above problem can be solved by a combination of an optimization routine and Algorithm \ref{alg:fw_back_sweep}. The optimization routine starts from an initial guess for the seed vector $\boldsymbol i_0=(i_{01},...,i_{0M})^T$ and keeps refining the value till the maximum is reached. Note that the $j$th element of $\tilde{\boldsymbol i}_0$ is now given by $\tilde i_{0j} = \sum_{m=1}^M ( \mathbbm 1_{\{j\in \mathbb N_m\}} i_{0m})$. As stated earlier, each node is selected as seed with probability $i_{0m}$ in group $\mathbb N_m$. For any value of the seed variable, the values of controls computed by Algorithm \ref{alg:fw_back_sweep} are such that Pontryagin's Principle is satisfied and hence they are optimum.

Numerical solution to the joint problem requires the gradient of the net reward function with respect to the seed variable vector $\boldsymbol{i_0}$ of length $M$. It is unlikely that the gradient can be computed analytically. A typical numerical optimization routine estimates it by perturbing the variable vector once in every dimension and computing the objective function at the perturbed point. Let the maximum number of optimization iterations be fixed to $L$. Then, in the worst case, Algorithm \ref{alg:fw_back_sweep}---which solves a system of $N$ differential equations $P$ times---has to be used $(M+1)L$ times to get a solution ($M$ evaluations to compute the gradient and one to compute the objective at the new point). That is, $N$ differential equations has to be solved $P(M+1)L$ times which has a complexity of $O(MN)$.

\section{Fixed Budget Constraint}
\label{sec:fixed_budget}

Our optimization framework and numerical scheme are flexible enough to handle other similar problems. In this section we discuss a variant of Problem (\ref{eq:opt_prob}) which maximizes the fraction of infected nodes in the network under a resource budget constraint. For this formulation we consider the seeds to be given and not as optimization variables. The problem is stated as:
\begin{subequations}
\label{eq:opt_prob_budget}
\begin{align}
\underset{ \begin{smallmatrix} (u_1(t),...,u_M(t))^T \end{smallmatrix} }{\text{max}} J & = \frac{1}{N}\sum_{j=1}^N i_j(T), \label{eq:cost_funtion_budget}\\
\text{subject to:~~} & \sum_{m=1}^M \int_0^Tg_m(u_m(t))dt = B, \label{eq:budget_constraint} \\
& (\ref{eq:opt_prob_states}), (\ref{eq:opt_prob_init_cond}). \end{align}
\end{subequations}

We can relax the budget constraint in Problem (\ref{eq:opt_prob_budget}) in the objective function using a multiplier---a technique commonly used in optimization theory---which leads to:
\begin{subequations}
\label{eq:opt_prob_budget_relax}
\begin{align}
\underset{ \begin{smallmatrix} \boldsymbol u \end{smallmatrix} }{\text{max}} ~~J & = \frac{1}{N}\sum_{j=1}^N i_j(T) - \mu \bigg( \sum_{m=1}^M \int_0^Tg_m(u_m(t))dt - B \bigg), \label{eq:cost_funtion_budget_relax}\\
\text{s.t.:~~} & (\ref{eq:opt_prob_states}), (\ref{eq:opt_prob_init_cond}). \end{align}
\end{subequations}
For the value of the multiplier which leads to Eq. (\ref{eq:budget_constraint}) being satisfied, Problem (\ref{eq:opt_prob_budget_relax}) solves Problem (\ref{eq:opt_prob_budget}). Also, given $\mu$, $\mu B$ is just a constant and can be removed from (\ref{eq:cost_funtion_budget_relax}).

Pontryagin's Principle applied to Problem (\ref{eq:opt_prob_budget_relax}) again leads to Eq. (\ref{eq:opt_prob_states}) as state equations, Eq. (\ref{eq:adjoint_eq}) as adjoint equations and Eq. (\ref{eq:transversaility_cond}) as transversality conditions. Only the Hamiltonian maximizing condition changes to:
\begin{equation}
u_m(t) = g_m'^{-1}\Big(\frac{1}{\mu^*}\sum_{l\in\mathbb N_m}\lambda_l^*(t) s_l^*(t)\Big),~1\leq m \leq M, \label{eq:hamil_max_cond_control_val_budget}
\end{equation}
where, $\mu^*$ is the value of multiplier at the optimum (which leads to Eq. (\ref{eq:budget_constraint}) being satisfied).

The above problem can be easily solved numerically by narrowing down to the value of the multiplier which satisfies Eq. (\ref{eq:budget_constraint}) by bisection search. The inner loop uses Algorithm \ref{alg:fw_back_sweep} (with (\ref{eq:hamil_max_cond_control_val1}) replaced by (\ref{eq:hamil_max_cond_control_val_budget})) to compute the controls. The outer loop adjusts the multiplier value by the bisection algorithm (initialized by appropriate guesses for minimum and maximum values for $\mu^*$) till the budget constraint (\ref{eq:budget_constraint}) is satisfied with desired accuracy. If the bisection search takes a maximum of $Q$ iterations to reach desired accuracy, Algorithm \ref{alg:fw_back_sweep} has to be used $Q$ times in the worst case\footnote{If $\mu_{high},\mu_{low}$ and $\mu_{th}$ are the initial guesses and the desired accuracy for $\mu$ in the bisection algorithm, we can compute $Q$ from the following: $(1/2)^Q(\mu_{high}-\mu_{low})\leq\mu_{th}\Rightarrow Q = \left\lceil \log_2((\mu_{high}-\mu_{low})/\mu_{th}) \right\rceil$. Here, $\lceil x \rceil$ is the least integer greater than or equal to $x$.}. Thus, in the worst case, a system of $N$ differential equations has to be solved $PQ$ times which has a computational complexity of $O(N)$.

\section{Node Centrality Measures to Form $M$ groups}
\label{sec:centralities}

We have used four commonly used notions of centralities in a network for dividing the individuals into $M$ groups. We briefly describe these notions here. Detailed discussions can be found in, for example, \cite[Chap. 7]{newman2009networks}.

\emph{Closeness Centrality:} Let $d_{ij}$ represent the length of a geodesic path from node $i$ to $j$---shortest path from $i$ to $j$ through the network. Then, $l_i \triangleq \sum_{j\in\mathbb N}d_{ij}/N$, takes small values for nodes which have smaller geodesic paths to other nodes in a network with a single giant component. Closeness centrality for node $i$ is $C_i \triangleq 1/l_i$. The nodes with larger values of closeness centrality are, on the average, at smaller distance from other nodes in the network. Thus, these nodes are potentially better spreaders.

\emph{Betweenness Centrality:} Let $n_{pq}^i=1$ if node $i$ lies in a geodesic path between nodes $p$ and $q$. Then betweenness centrality of node $i$ is $B_i \triangleq \sum_{p,q\in\mathbb N}n_{pq}^i$. It measures the extent to which node $i$ lies in the shortest paths between all pairs of nodes in the network. Therefore betweenness centrality can identify good spreaders in the network.

\emph{Degree Centrality:} The degree of a node in the network---number of connections to other nodes---is also known as degree centrality of the node. Intuitively, the nodes with more connections are better spreaders.

\emph{Pagerank Centrality:} The nodes connected to more central nodes in the network should have higher centrality. Let $P_i$ represent the pagerank centrality of node $i$. Then, the pagerank centralities of the nodes in the network satisfy the fixed point equations: $P_i = \eta \sum\limits_{j\in\mathbb N}A_{ij}P_j/k_j+\delta$, $1\leq i\leq N$. Here, $k_j$ is the degree of node $j$. We choose $\eta=0.85$ and $\delta=1$ in this paper. We choose pagerank as one of the measures because it generalizes other measures like degree centrality and eigenvector centrality\footnote{Eigenvector centrality of node $i$ is equal to the $i$th entry in the eigenvector corresponding to the largest eigenvalue of the adjacency matrix.} and has found widespread use in ranking pages on the world wide web \cite[Chap. 7]{newman2009networks}.

\section{Results}
\label{sec:results}

\emph{Network and Default Model Parameters:} We present our results on networks sampled from giant components of three real world social networks: `condensed matter archive' scientific collaboration network, the undirected version of Slashdot network, and Facebook network \cite{snap}. The (sampled) networks are of sizes $12000$, $12000$ and $4020$ nodes respectively. We sample the networks from a random starting node using breadth first traversal. These networks provide the adjacency matrices $\boldsymbol A$ used to demonstrate results in this section.\footnote{The original scientific collaboration, Slashdot and Facebook networks have $23133$, $77350$, and $4039$ nodes respectively. They contain $93497$, $516575$ and $88234$ edges, and have a diameter of $14$, $11$ and $8$ respectively \cite{snap}.}

The spreading rate and the deadline together decide the extent of spreading in the SI process. Hence, we have fixed the deadline to $T=1$ time unit for all the results presented in this paper, and have varied the spreading rate in the parameter sweep studies. We choose the default values for the spreading rates to be $\beta=0.1,0.04,0.035$ for the scientific collaboration, Slashdot and Facebook networks respectively. These values are chosen because they lead to small fractions of infected nodes at the deadline---$0.1007,0.1050,0.1048$---in the respective uncontrolled systems. Such a system, where there is scope to reach more individuals, will benefit from campaigning.

The cost function for group $m$ in (\ref{eq:cost_funtion}) is chosen to be, $g_m(u_m(t))=bp_mu_m^2(t)$. Here, $p_m=|\mathbb N_m|/N$ as explained earlier is the fraction of nodes in group $m$. The factor $b$ weighs the reward due to the final fraction of infected nodes and the cost of application of controls in the net reward function (\ref{eq:cost_funtion}). The default value is set to $b=25$ for all the three networks.

\emph{Division into Groups:} We have divided the whole network into $M$ equal groups for the purpose of applying controls. Each group is influenced by a separate control function. Thus, $p_m=1/M,~\forall m$. The groups were created based on the four centrality measures discussed in Sec. \ref{sec:centralities}. The values of the net reward achieved by following the four grouping strategies identify the better-performing centrality measures. We calculate node centralities for all nodes in the network. The last $N/M$ nodes with the least value of the centrality are placed in group $1$, the next $N/M$ nodes are placed in group $2$, ... , and the top $N/M$ nodes are placed in group $M$.

\emph{Heuristic Controls:} We compare the results obtained by the optimal control strategies with two heuristic strategies. The first one is the \emph{\textbf{best} static/constant control}. It is constant over time throughout the campaign duration and maximizes the net reward in (\ref{eq:cost_funtion}). The same control is applied to all nodes (or groups) and the seed is uniformly distributed throughout the network. Comparison of the optimal strategy with the best constant control reveals the improvement due to dynamic resource allocation over time and heterogeneous allocation over groups.

The second heuristic control is the \emph{\textbf{best} two-stage control}. It is a simple dynamic control which is constant in $t\in[0,T/2]$ and $0$ for the rest of the campaign duration. The value of the non-zero part is selected so that the reward in (\ref{eq:cost_funtion}) is maximized. Again, seeds are selected uniformly from the population and the same control is applied to all groups. Theorem \ref{thm:control_noninc} suggests that the initial period of the campaign is more important and requires stronger control, which motivates this heuristic strategy.

Since both the heuristic controls are optimized controls in a restricted class of functions, an optimization routine calculates the heuristic controls numerically. If the maximum number of optimization iterations is fixed to $S$, in the worst case, computation of the heuristic strategies requires solving a system of $N$ differential equations $S$ times. Thus, the worst case computational complexity is $O(N)$.

\subsection{Shapes of the Controls and Importance of the Groups (Given Seed)}

\begin{figure}[ht!]
\centering % 55mm
\includegraphics[width=60mm]{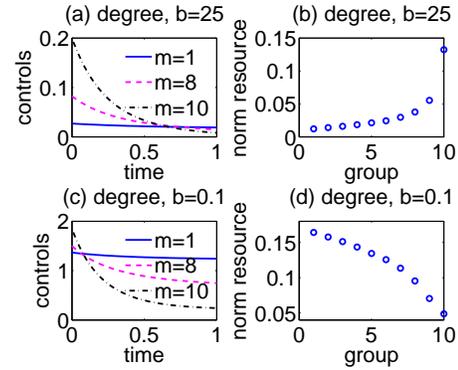}
\caption{\small{Shapes of controls and normalized (or per capita) resource consumption for the \textbf{scientific collaboration network} for degree centrality. Parameter values: $\beta=0.1,~i_{0m}=0.01\forall m$. For (a) and (b) resource is scarce ($b=25$), (c) and (d) shows abundant resource case ($b=0.1$).}}
\label{fig:norm_res_b_25_and_b_pt1}
\end{figure}

\begin{figure}[ht!]
\centering % 55mm
\includegraphics[width=60mm]{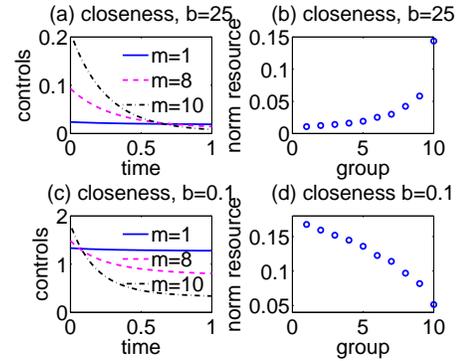}
\caption{\small{Shapes of controls and normalized (or per capita) resource consumption for the \textbf{Slashdot} network for closeness centrality. Parameter values: $\beta=0.04,~i_{0m}=0.01\forall m$. For (a) and (b) resource is scarce ($b=25$), for (c) and (d) resource is abundant ($b=0.1$).}}
\label{fig:norm_res_b_25_and_b_pt1_slashdot}
\end{figure}

\subsubsection{Shapes of the Controls}

We plot shapes of the controls and normalized resource allocation for the scientific collaboration and Slashdot networks in Figs. \ref{fig:norm_res_b_25_and_b_pt1} and \ref{fig:norm_res_b_25_and_b_pt1_slashdot} respectively. Here the networks are divided into $M=10$ groups. These plots are solutions of Problem (\ref{eq:opt_prob}) for the case in which seeds are uniformly selected from the population (and are not optimization variables). We have shown results for only one centrality measure in both the cases, other centrality measures lead to same qualitative trends hence plots are removed for brevity.

Figs. \ref{fig:norm_res_b_25_and_b_pt1}a and \ref{fig:norm_res_b_25_and_b_pt1_slashdot}a plots the controls for three representative groups for the scarce resource case ($b=25$), and Figs. \ref{fig:norm_res_b_25_and_b_pt1}c and \ref{fig:norm_res_b_25_and_b_pt1_slashdot}c for the abundant resource case ($b=0.1$).  The figures show that the initial period is more important for campaigning in a susceptible-infected epidemic---early infection of nodes gives them more time to infect others, hence this behavior.

The resource allocation is more variable over time for the group containing more central nodes ($m=10$) than the group containing less central nodes ($m=1$) for both the cases---when the resource is scarce ($b=25$, in Figs. \ref{fig:norm_res_b_25_and_b_pt1}a, \ref{fig:norm_res_b_25_and_b_pt1_slashdot}a) and when it is abundant ($b=0.1$, in Figs. \ref{fig:norm_res_b_25_and_b_pt1}c, \ref{fig:norm_res_b_25_and_b_pt1_slashdot}c). Notice that for the convex-increasing (specifically quadratic) instantaneous cost of application of the control considered in this work, a uniform shape over time has more aggregate effort than any non-uniform shape for a given amount of resource\footnote{Consider the following example: Suppose decision horizon is discretized into three intervals and the control strength is constant within the interval. A control profile with $0.8,0.4,0.1$ units of effort in the three intervals uses $0.81 c$ units of resource (for a constant $c$), and has an aggregate effort of $1.3$ units over the three intervals. However, a uniform profile which uses same amount of resource has an effort of $\approx 0.5196$ units in each interval, \emph{i.e.,} an aggregate effort of $\approx 1.5588$ units over the three intervals. This is $\approx 20\%$ greater than the aggregate effort in the previous case.}. This is so because the per unit cost of applying extra control strength grows super-linearly at any point. In-spite of this, controls are sharply decreasing for the high centrality group ($m=10$). This is because central nodes are better spreaders, so it is useful to target them early---even though it may be costly---because it will lead to further infection.

In contrast, controls are more uniform over time for low centrality nodes. They are not only poor spreaders but are also disadvantaged in receiving the message due to their non-central positions. The focus of the optimal strategy for these nodes is to apply the maximum possible control effort (by keeping the control signal uniform) so that maximum direct recruitment---due to more aggregate effort---can be achieved. There is not much benefit from early infection of these low centrality nodes.

\subsubsection{Importance of the Groups}

To identify the important groups for campaigning, we plot normalized resource consumed over the complete campaign horizon by the $M$ groups. Figs. \ref{fig:norm_res_b_25_and_b_pt1}b, \ref{fig:norm_res_b_25_and_b_pt1_slashdot}b correspond to the scarce resource case ($b=25$) for the two networks. Figs. \ref{fig:norm_res_b_25_and_b_pt1}d, \ref{fig:norm_res_b_25_and_b_pt1_slashdot}d correspond to the abundant resource case ($b=0.1$). The normalized (or per capita) resource consumed by group $m$ is calculated as $r_{per~capita}^m = (1/p_m)\int_0^Tg_m(u_m(t))dt=b\int_0^Tu_m^2(t)dt$.

From Figs. \ref{fig:norm_res_b_25_and_b_pt1}b and \ref{fig:norm_res_b_25_and_b_pt1_slashdot}b we conclude that in the scarce resource case ($b=25$), the optimal strategy targets the groups with higher centrality measures. They are better spreaders in the population and this leads to the best utilization of the available resources because their infection leads to further information spread through epidemic process.

However, when resource is abundant, $b=0.1$ (Figs. \ref{fig:norm_res_b_25_and_b_pt1}d, \ref{fig:norm_res_b_25_and_b_pt1_slashdot}d), groups with lower centrality measures get more allocation. In this case, we have enough resources, and targeting the groups which are at a disadvantage in receiving the message through epidemic process is a better strategy. Note that more central nodes are better positioned to receive the message and hence receive relatively less allocation when the resource is abundant.

\begin{figure}[ht!]
\centering % 55mm
\includegraphics[width=60mm]{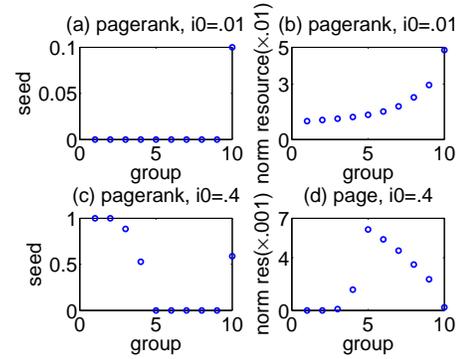}
\caption{\small{Fraction of seed in each group and normalized (or per capita) resource consumption for the \textbf{Slashdot network} for pagerank centrality. Parameter values: $b=25,~\beta=0.04$. (a) and (b) have small seed budget, (c) and (d) have large seed budget.}}
\label{fig:seed_i0_01_and_i0_4_b_25_slashdot}
\end{figure}

\subsection{Seed and Resource Allocation by the Joint Problem}

For the joint Problem (\ref{eq:opt_prob}) in which both the dynamic resource allocation and seeds for the epidemic are optimized, we plot the fractions of population selected as seeds and the normalized resource allocated to each group in Fig. \ref{fig:seed_i0_01_and_i0_4_b_25_slashdot} for the Slashdot network for the pagerank centrality measure. We choose number of groups, $M=10$ and consider two cases---when the seed budget is low, $i_0=0.01$, (Figs. \ref{fig:seed_i0_01_and_i0_4_b_25_slashdot}a, \ref{fig:seed_i0_01_and_i0_4_b_25_slashdot}b) and when the seed budget is high, $i_0=0.4$, (Figs. \ref{fig:seed_i0_01_and_i0_4_b_25_slashdot}c, \ref{fig:seed_i0_01_and_i0_4_b_25_slashdot}d). Plots for other centrality measures and networks show similar trend and are omitted for brevity.

We see a behavior similar to that in the previous section in the allocation. Central nodes are selected as seeds when the seed budget is low, owing to there spreading strength. In contrast, for the high seed budget case, primarily disadvantaged nodes are selected because epidemic spreading cannot be relied upon to inform them. However, when seed budget is high, some central nodes are also selected as seeds because the optimal strategy wants some good spreaders in the population. For high seed budget case, controls too target the non-central nodes (groups which are not already completely seeded) as seen in Fig. \ref{fig:seed_i0_01_and_i0_4_b_25_slashdot}d.

\begin{figure}[ht!]
\centering % 78mm
\hspace{-1em}
\includegraphics[width=88mm]{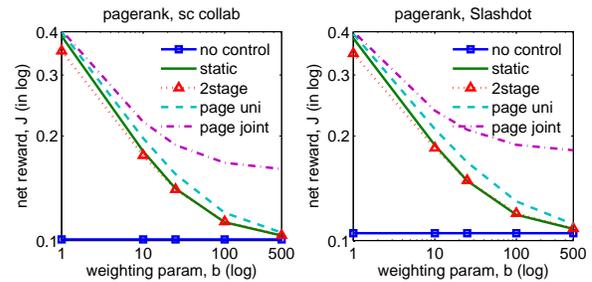}
\caption{\small{Reward $J$ vs weighting parameter $b$ for pagerank centrality.}}
\label{fig:J_vs_b_pagerank_all_nw}
\end{figure}

\begin{figure*}[ht!]
\centering % 50mm
\subfloat[Scientific collaboration network. \label{fig:J_vs_b_percent_improve}]{
\includegraphics[width=57mm]{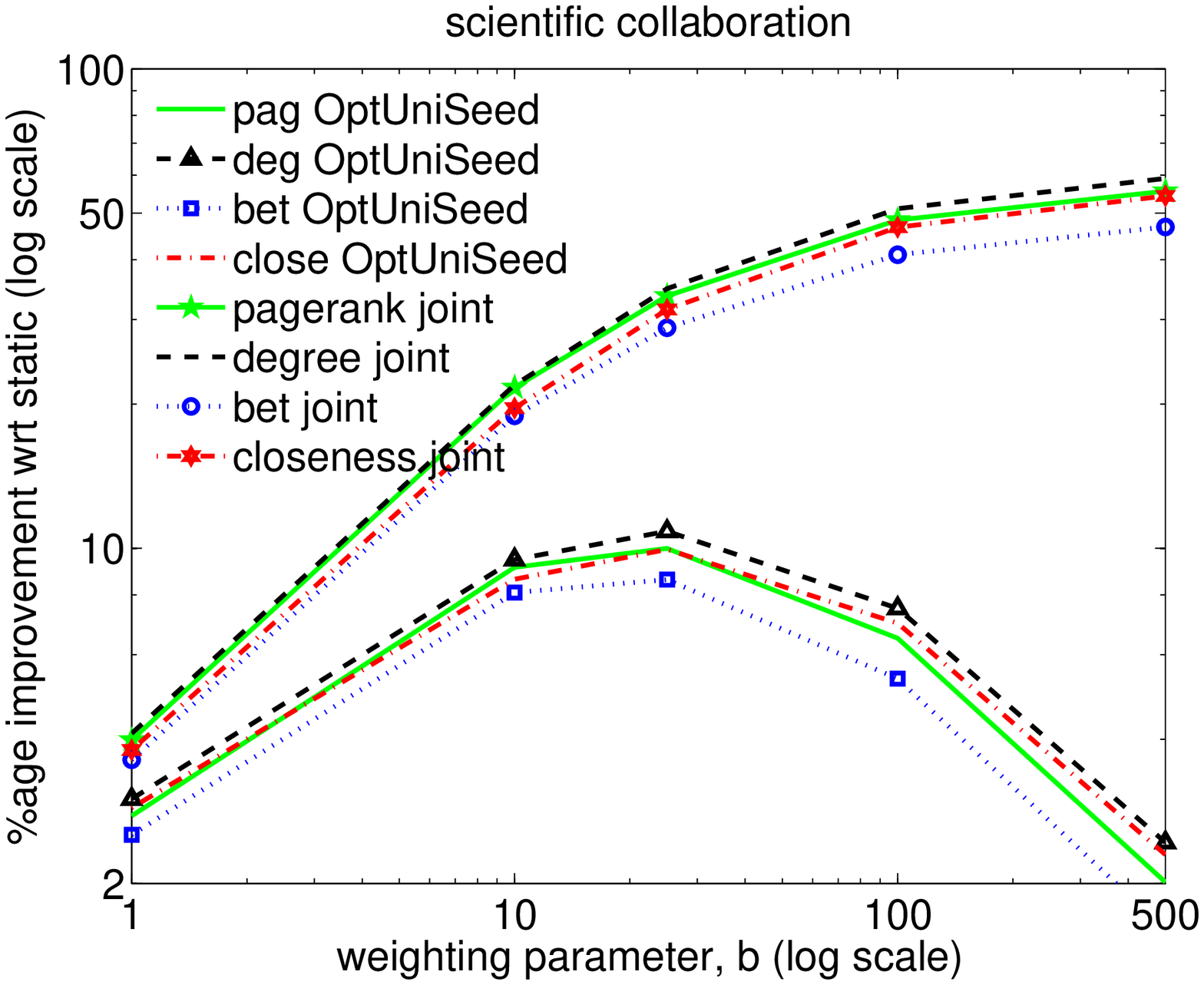} }
\hfill
\subfloat[Slashdot network. \label{fig:J_vs_b_percent_improve_slashdot}]{
\includegraphics[width=57mm]{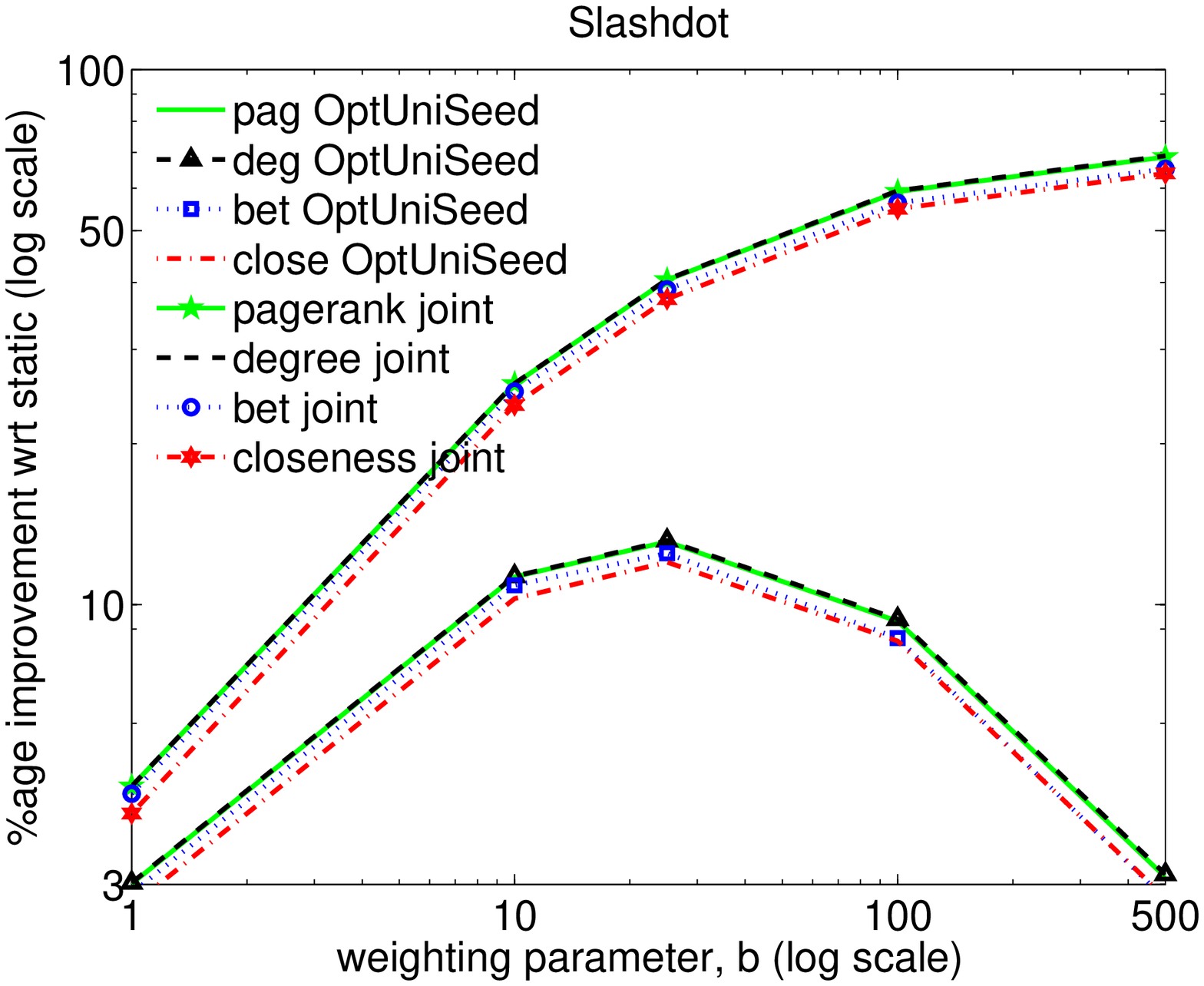} }
\hfill
\subfloat[Facebook network. \label{fig:J_vs_b_percent_improve_facebook}]{
\includegraphics[width=57mm]{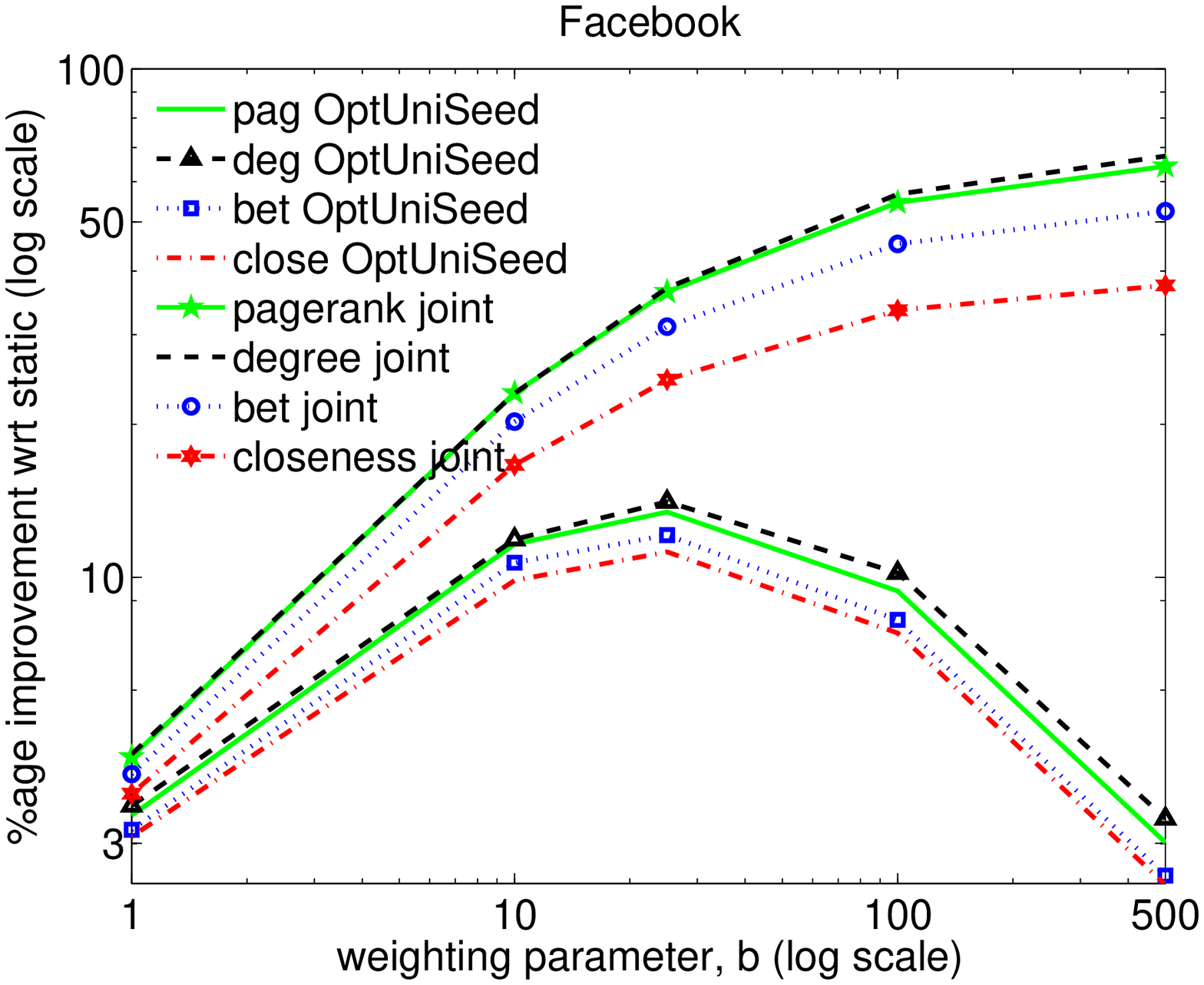} }
\caption{\small{Variation of percentage improvements in the net reward achieved by various optimal strategies over the static control strategy, with respect to weighting parameter $b$. Larger $b\Rightarrow$ costly resource. Parameter values for scientific collaboration network: $\beta=0.1,~i_{0}=0.01,~M=5$; Slashdot: $\beta=0.04,~i_{0}=0.01,~M=5$; Facebook:  $\beta=0.035,~i_{0}=0.01,~M=5$.}}
\label{fig:J_vs_b}
\end{figure*}

\subsection{Effect of Model Parameters on the Net Reward, and Identifying Better Centrality Measures}

We study the effect of varying three system parameters---the cost of application of controls $b$, the spreading rate $\beta$, and the number of groups $M$---on the scientific collaboration, Slashdot and Facebook networks in this section.

\subsubsection{Varying the Cost of Application of Controls}
\label{sec:results_J_vs_b}

Fig. \ref{fig:J_vs_b_pagerank_all_nw} shows the variation of the net reward function, $J$ with the cost of application of controls---captured by varying the parameter $b$ in the instantaneous cost function $g_m(u_m(t))=bu_m^2(t)$---for the pagerank centrality measure for the scientific collaboration and Slashdot networks. We present results for both the cases---when only time varying resource is optimized (seeds are uniformly selected from the population), and for the joint optimization of the seeds and dynamic resource allocation. We compare the performance of the two optimal strategies with the two heuristic strategies. Other networks and centrality measures have similar trend, so plots are omitted for brevity. Figs. \ref{fig:J_vs_b_percent_improve}, \ref{fig:J_vs_b_percent_improve_slashdot}, \ref{fig:J_vs_b_percent_improve_facebook} plots the percentage improvement achieved by optimal strategies over the static strategy for the three networks for all four centrality measures.

Fig. \ref{fig:J_vs_b_pagerank_all_nw} shows that as resource becomes costly, and hence scarce, we are able to reach fewer individuals in the network (due to lack of resources only epidemic spreading disseminates information, optimal control does not help). We observe from the percentage improvement data in Fig. \ref{fig:J_vs_b} that optimal time varying resource allocation alone (without seed optimization) achieves substantial gains for a window of  intermediate values of $b$. If the resource is abundant (small $b$), even the static strategy reaches a lot of nodes---leading to a small percentage improvement in optimal strategy compared to the static strategy. On the other hand, scarce resource (large $b$) is not enough to influence many people, which explains such a trend. The joint seed--dynamic resource allocation achieves much higher gains. In addition to using the resource in the best possible way, the epidemic process is also enhanced due to seed optimization. This is particularly useful when the resource is scarce (high values of $b$).

It is worth noting that, node degree---in-spite of being a simple and local centrality measure---does not have any disadvantage over more complicated measures. This behavior is consistent for different networks (and also when other system parameters are varied in Figs. \ref{fig:J_vs_beta} and \ref{fig:J_vs_M}). This observation may be useful because node degree can be locally computed.\footnote{A possible reason why we observe such behavior is that epidemic spreading is local. The message spreads from a seed only to a few hops; thus degree---a local measure---is a better metric to identify better spreaders than betweenness or closeness which depend on all other nodes in the network, most of which are at large distances from the node in question.}

\begin{figure}[ht!]
\centering % 78mm
\hspace{-1em}
\includegraphics[width=88mm]{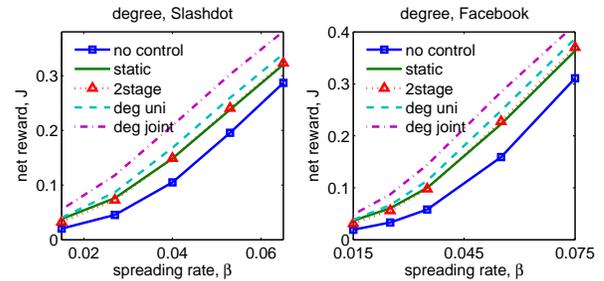}
\caption{\small{Reward $J$ vs spreading rate $\beta$ for degree centrality.}}
\label{fig:J_vs_beta_degree_all_nw}
\end{figure}

\begin{figure*}[ht!]
\centering % 50mm
\subfloat[Scientific collaboration network. \label{fig:J_vs_beta_percent_improve}]{
\includegraphics[width=57mm]{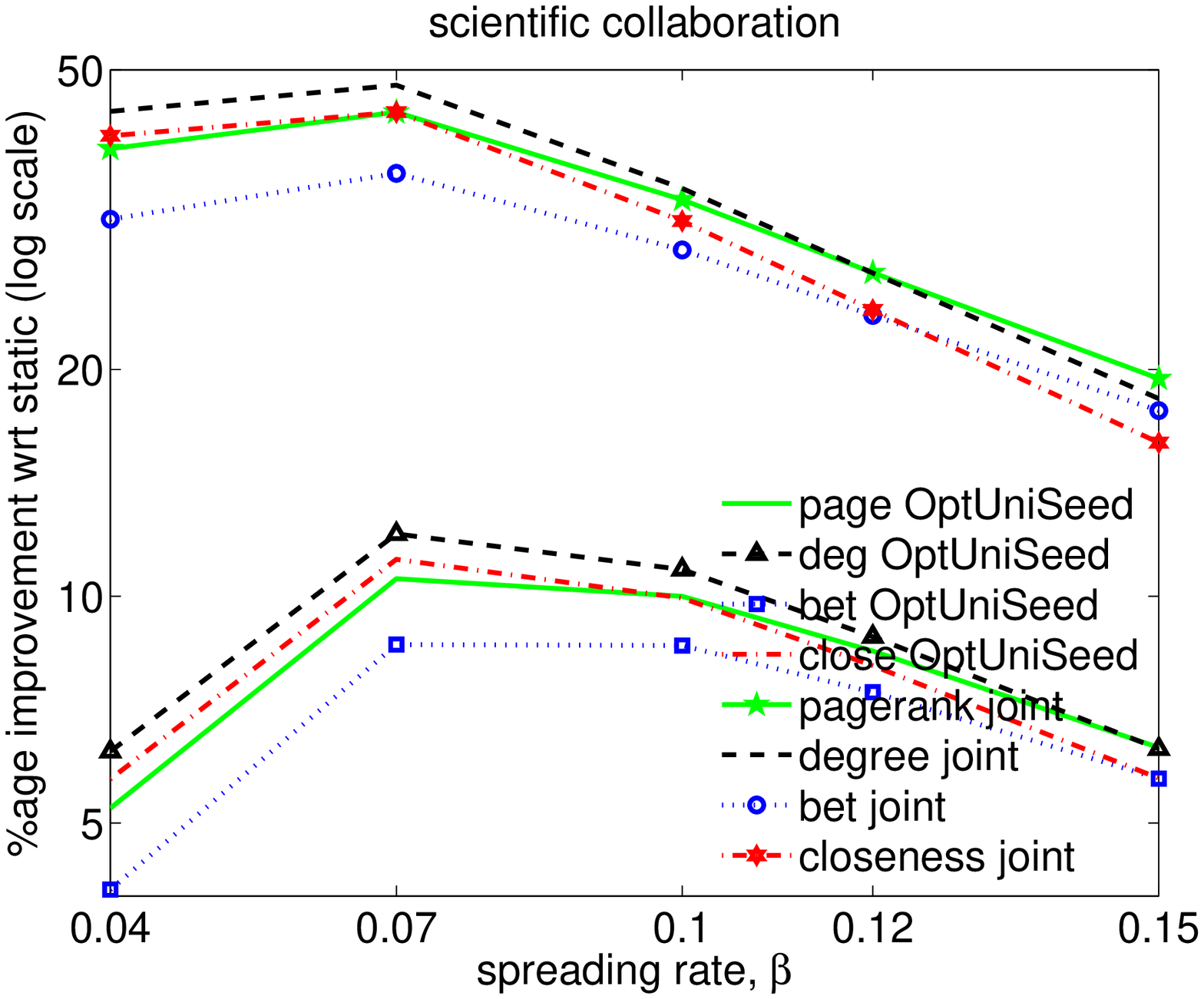} }
\hfill
\subfloat[Slashdot network. \label{fig:J_vs_beta_percent_improve_slashdot}]{
\includegraphics[width=57mm]{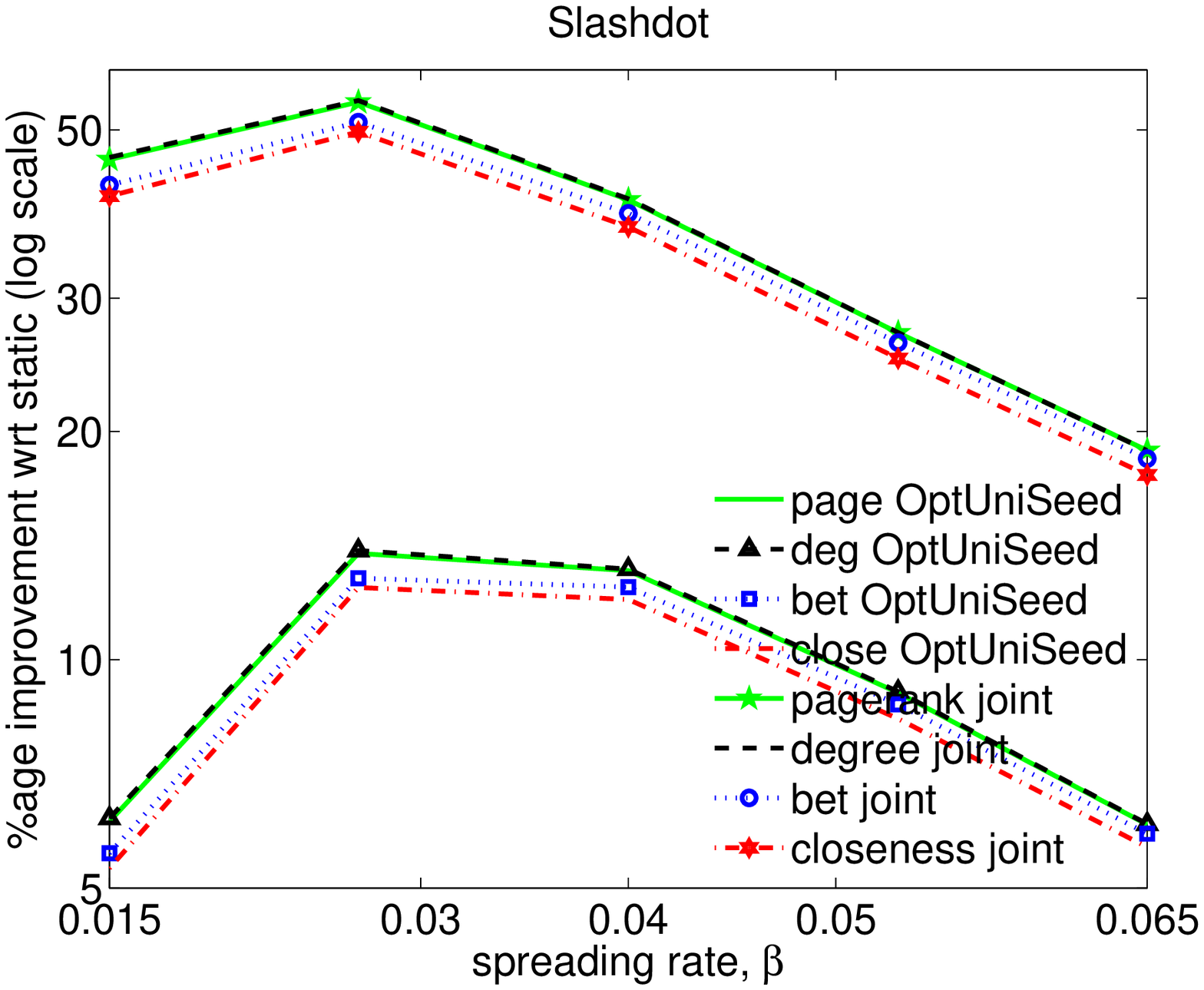} }
\hfill
\subfloat[Facebook network. \label{fig:J_vs_beta_percent_improve_facebook}]{
\includegraphics[width=57mm]{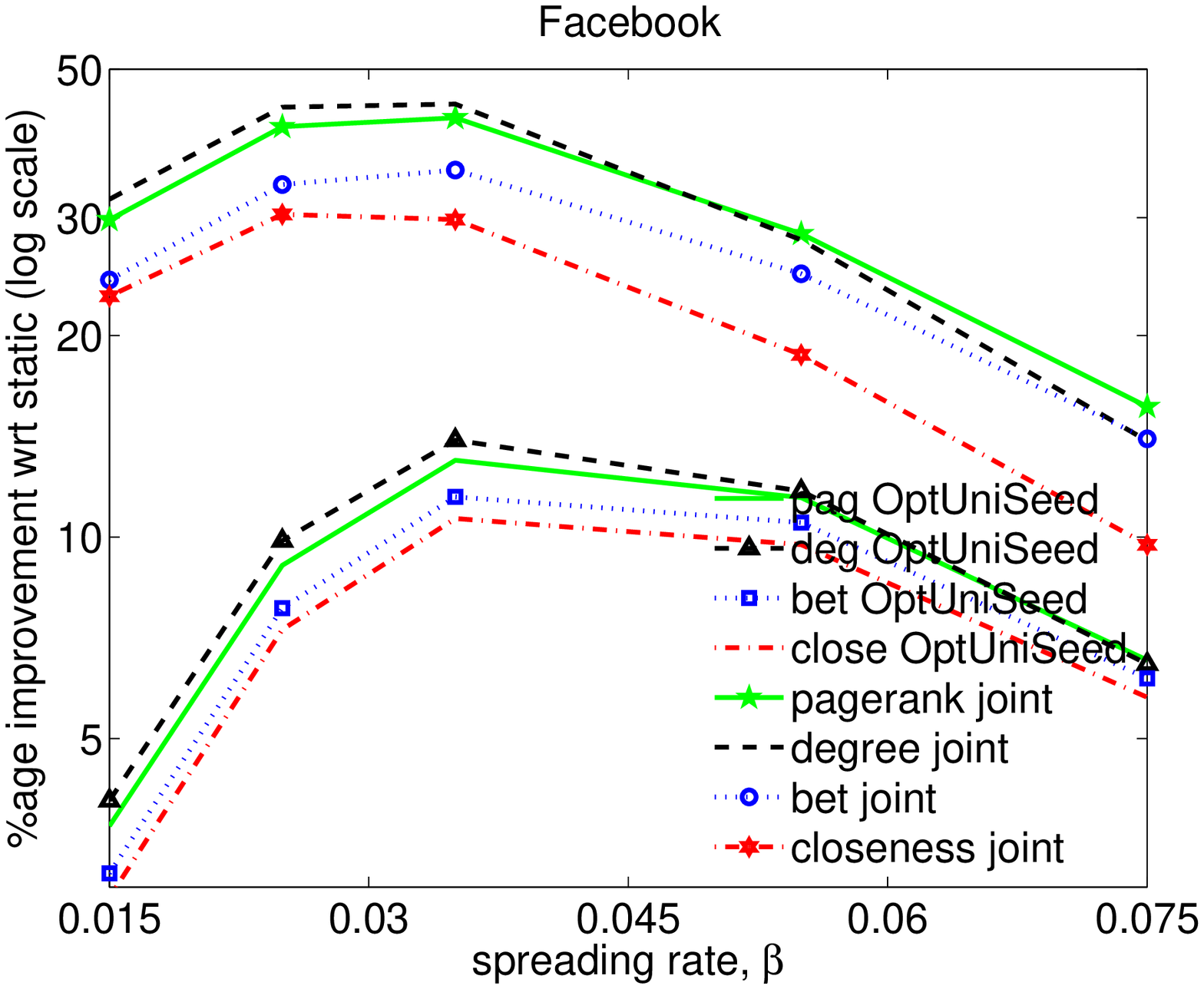} }
\caption{\small{Variation of percentage improvements in the net reward achieved by various optimal strategies over the static control strategy, with respect to information diffusion rate $\beta$. Parameter values for scientific collaboration network: $b=25,~i_{0}=0.01,~M=5$; Slashdot: $b=25,~i_{0}=0.01,~M=5$; Facebook:  $b=25,~i_{0}=0.01,~M=5$.}}
\label{fig:J_vs_beta}
\end{figure*}

\subsubsection{Varying the Spreading Rate, $\beta$}

Figs. \ref{fig:J_vs_beta_degree_all_nw} and \ref{fig:J_vs_beta} plot the performance of the optimal strategies with respect to the spreading rate $\beta$ for the three networks. Larger $\beta$---faster spreading---leads to more informed population at the deadline for any strategy followed. This is the reason for increasing trend in the curves in Fig. \ref{fig:J_vs_beta_degree_all_nw}. From the percentage improvement data (Fig. \ref{fig:J_vs_beta}), we observe that optimal strategies are more beneficial for a window of intermediate values of $\beta$. When $\beta$ is high, even heuristic strategies reach a lot of people and hence optimal strategies achieve only a low percentage improvement over them. When $\beta$ is low, the advantage of targeting good spreaders by direct recruitment is low due to slow spreading; this causes small percentage improvement in the optimal strategies over the static strategy.

\begin{figure}[ht!]
\centering % 78mm
\hspace{-1em}
\includegraphics[width=88mm]{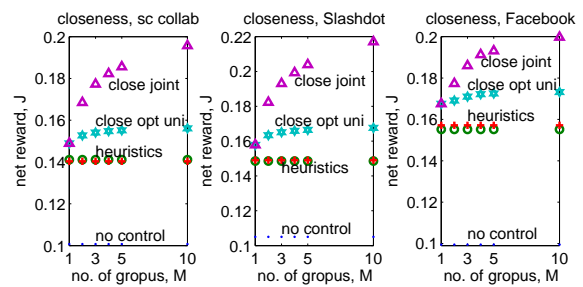}
\caption{\small{Reward $J$ vs number of groups $M$ for closeness centrality.}}
\label{fig:J_vs_M_close_all_nw}
\end{figure}

\begin{figure*}[ht!]
\centering % 50mm
\subfloat[Scientific collaboration network. \label{fig:J_vs_M_percent_improve}]{
\includegraphics[width=57mm]{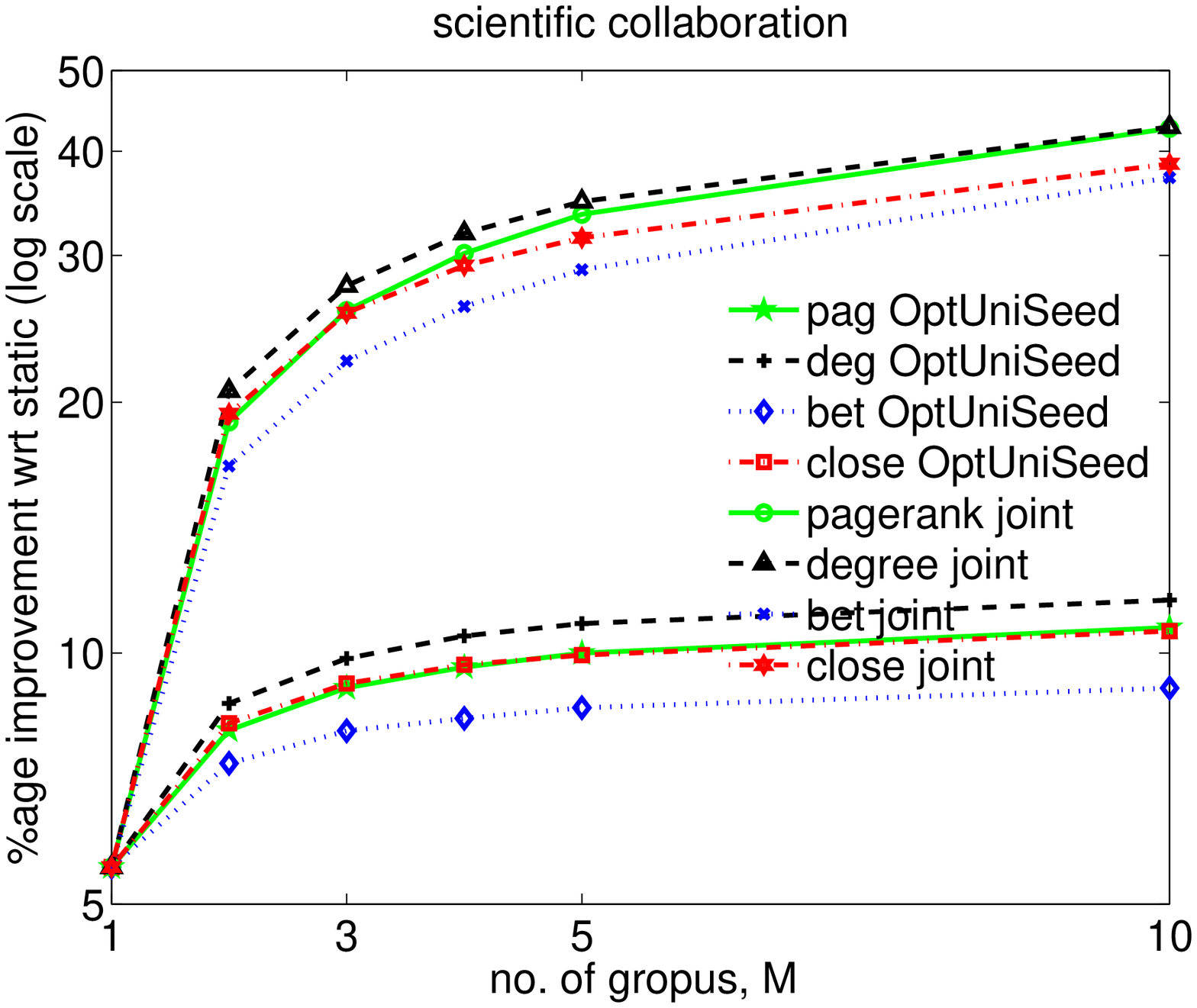} }
\hfill
\subfloat[Slashdot network. \label{fig:J_vs_M_percent_improve_slashdot}]{
\includegraphics[width=57mm]{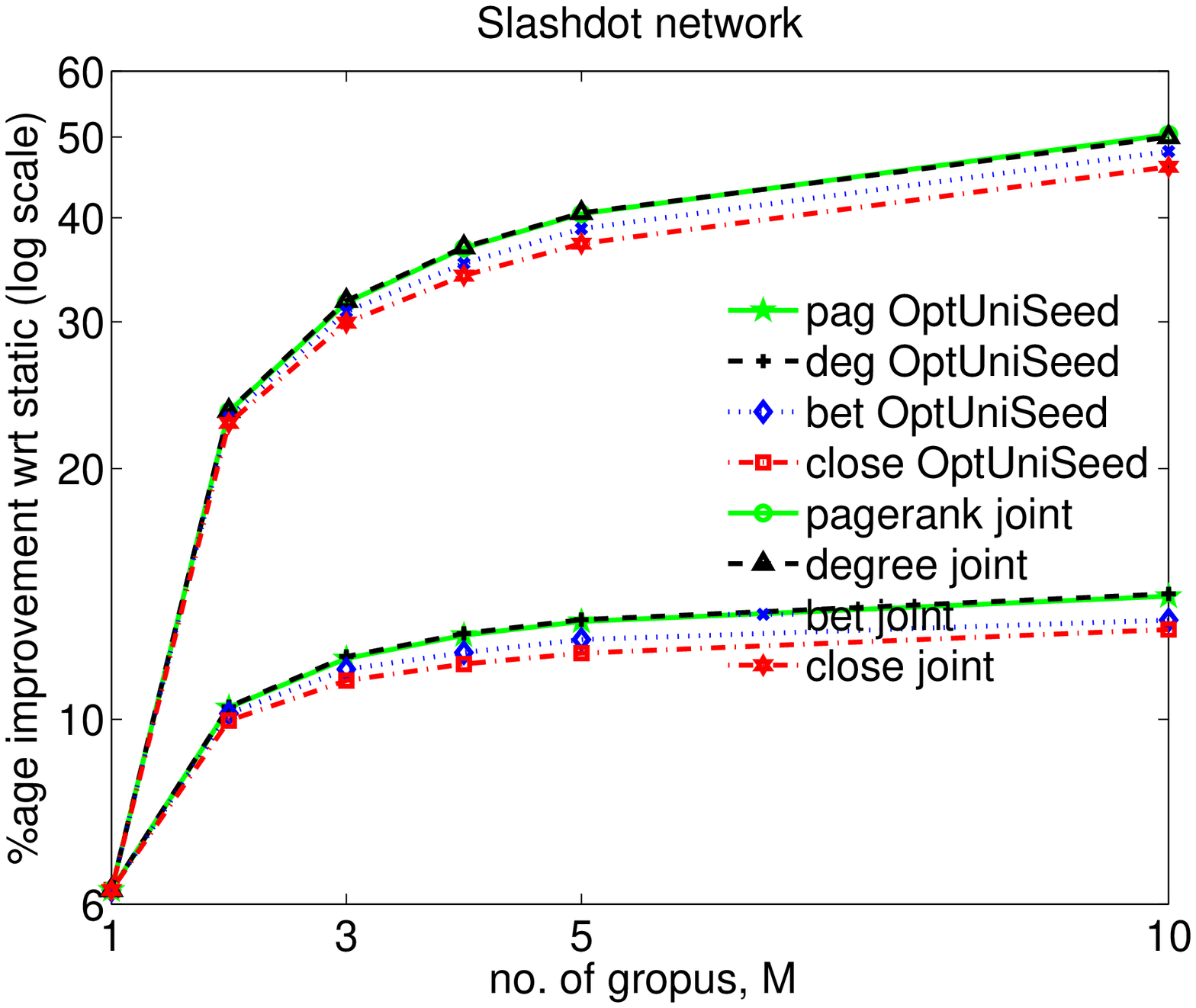} }
\hfill
\subfloat[Facebook network. \label{fig:J_vs_M_percent_improve_facebook}]{
\includegraphics[width=57mm]{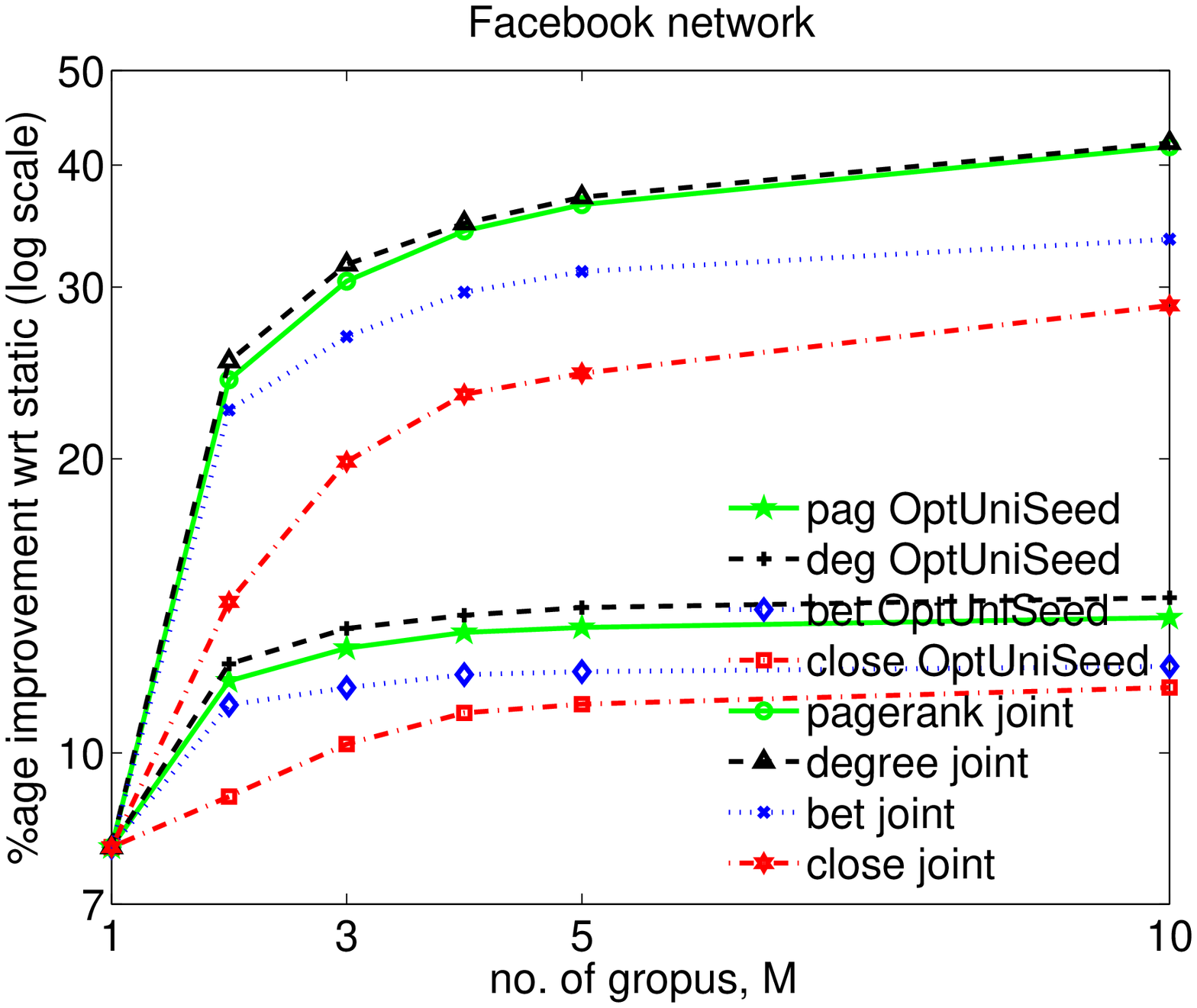} }
\caption{\small{Variation of percentage improvements in the net reward achieved by various optimal strategies over the static control strategy, with respect to number of groups $M$. Parameter values for scientific collaboration network: $b=25,~\beta=0.1,~i_{0}=0.01$; Slashdot: $b=25,~\beta=0.04,~i_{0}=0.01$; Facebook:  $b=25,~\beta=0.035,~i_{0}=0.01$.}}
\label{fig:J_vs_M}
\end{figure*}

\subsubsection{Varying the Number of Groups, $M$}

Figs. \ref{fig:J_vs_M_close_all_nw} and \ref{fig:J_vs_M} plot the impact of changing the number of groups into which nodes are divided for the purpose of applying controls. Finer division allows more flexible resource distribution, leading to a better reward value. For dynamic resource optimization without seed selection, we see a saturation in the percentage improvements at about $M=5$ for all three networks (Fig. \ref{fig:J_vs_M}) because $5$ groups seems to be enough to capture the disparity in centrality measures for the purpose of applying dynamic controls. However, joint allocation keeps achieving better results with increasing $M$. But the increase in percentage improvement due to addition of more groups is smaller at larger values of $M$.

\begin{figure}[ht!]
\centering % 78mm
\hspace{-1em}
\includegraphics[width=88mm]{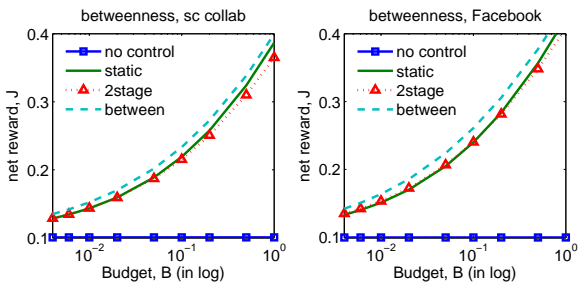}
\caption{\small{Reward $J$ vs budget $B$ for betweenness centrality.}}
\label{fig:J_vs_budget_between_all_nw}
\end{figure}

\begin{figure*}[ht!]
\centering % 50mm
\subfloat[Scientific collaboration network. \label{fig:J_vs_budget_percent_improve}]{
\includegraphics[width=57mm]{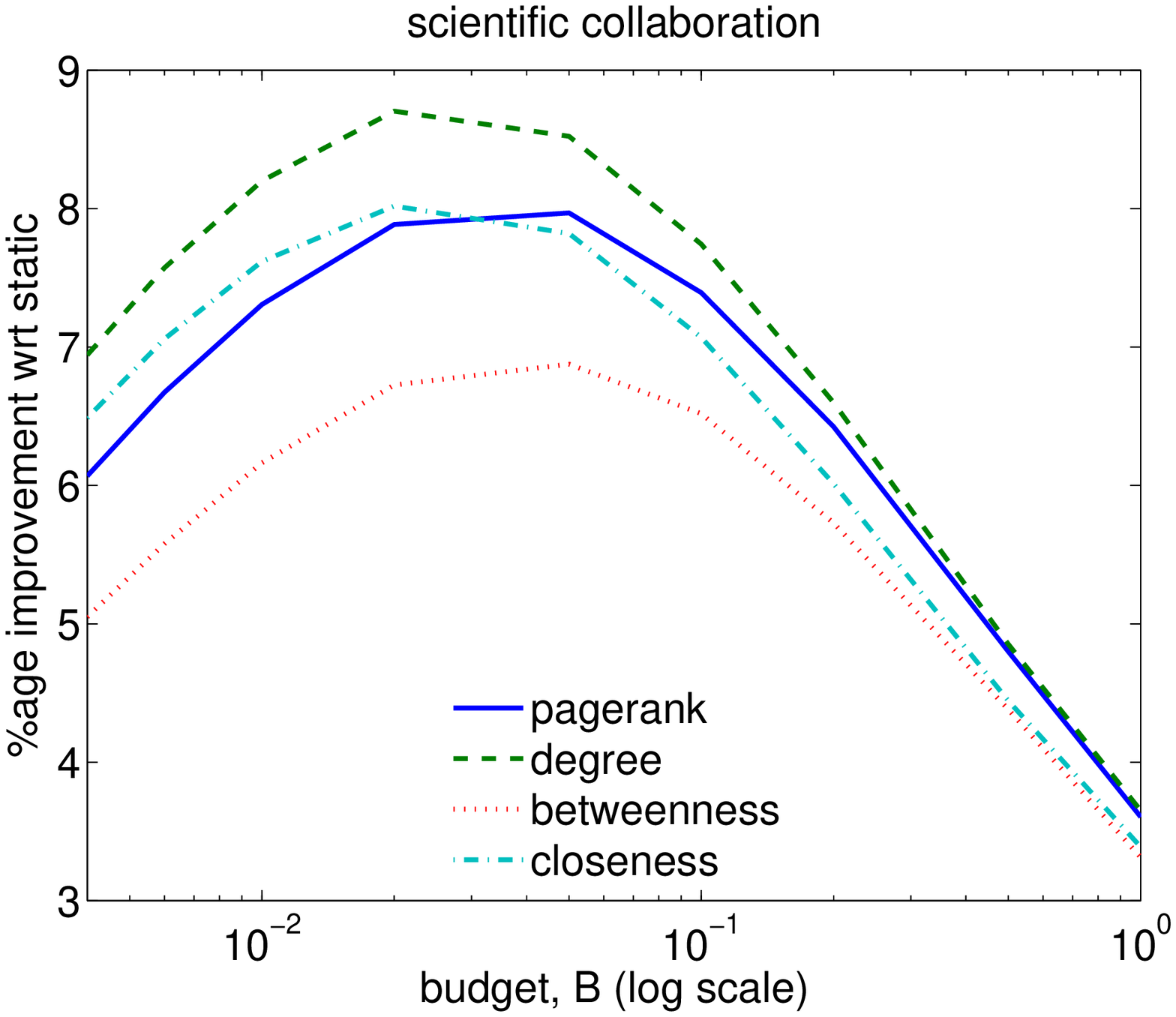} }
\hfill
\subfloat[Slashdot network. \label{fig:J_vs_budget_percent_improve_slashdot}]{
\includegraphics[width=57mm]{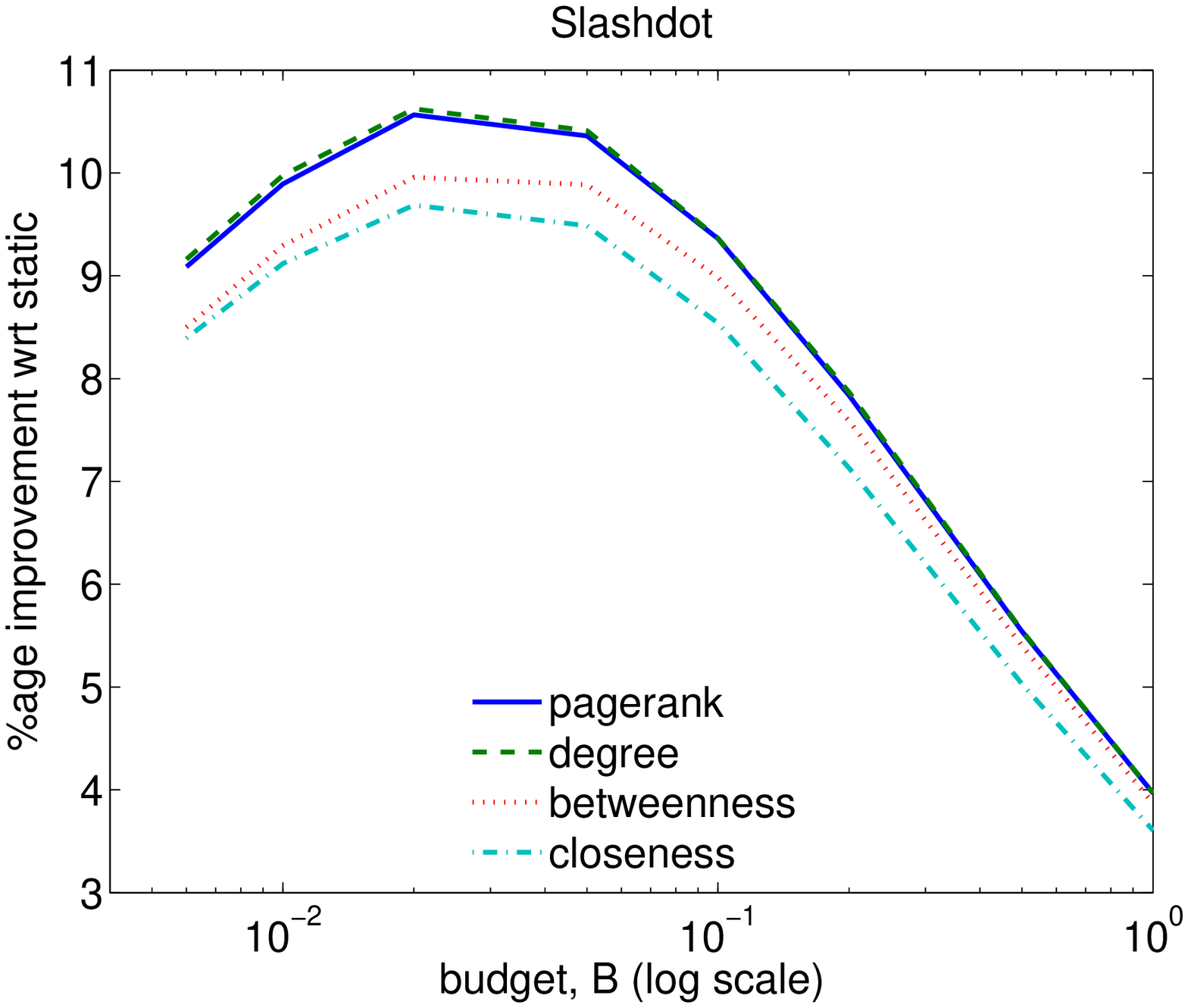} }
\hfill
\subfloat[Facebook network. \label{fig:J_vs_budget_percent_improve_facebook}]{
\includegraphics[width=57mm]{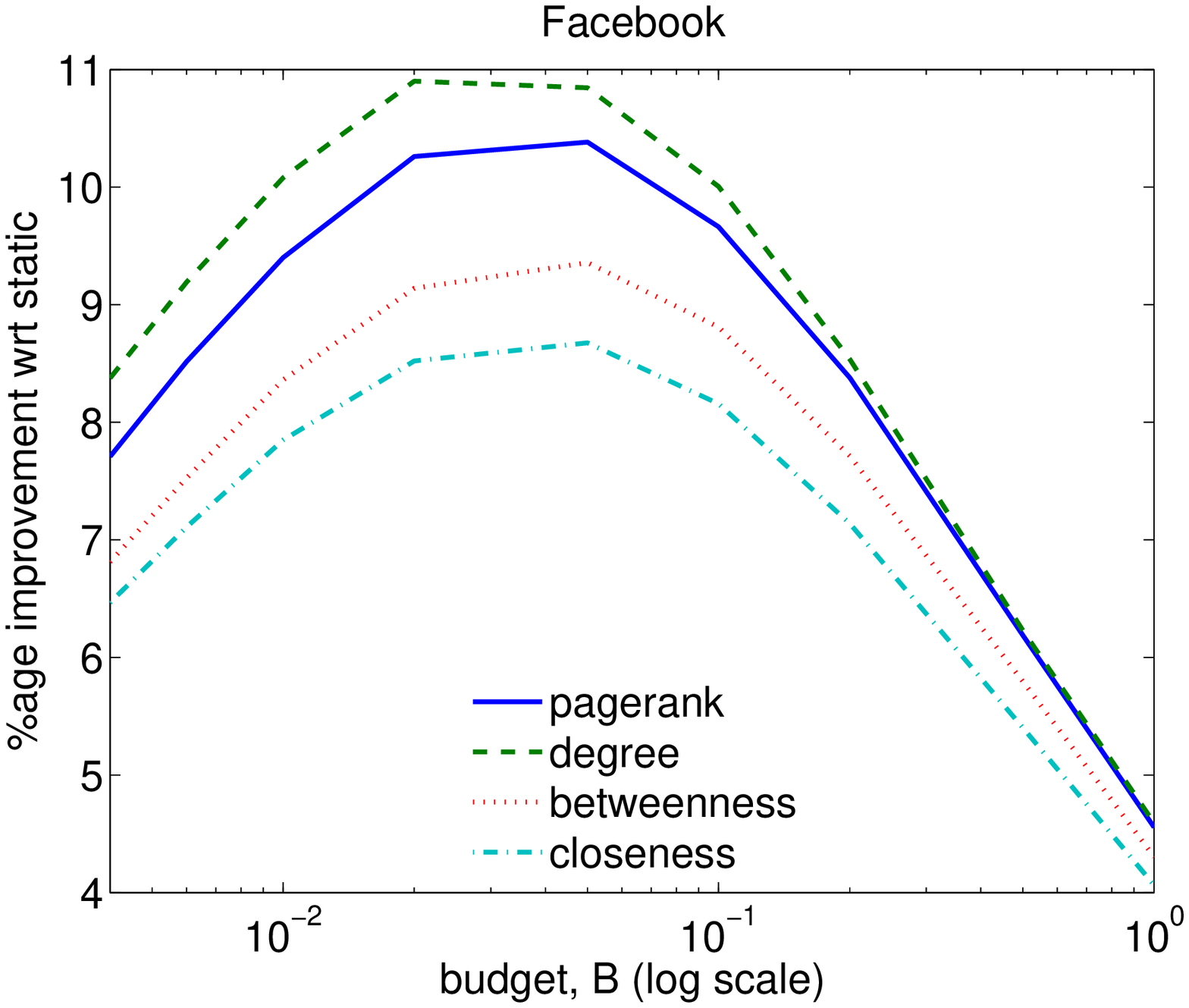} }
\caption{\small{Variation of percentage improvements in the net reward achieved by various optimal strategies over the static control strategy, with respect to budget $B$. Parameter values for scientific collaboration network: $b=25,~\beta=0.1,~i_{0}=0.01,~M=5$; Slashdot: $b=25,~\beta=0.04,~i_{0}=0.01,~M=5$; Facebook:  $b=25,~\beta=0.035,~i_{0}=0.01,~M=5$.}}
\label{fig:J_vs_budget}
\end{figure*}

\subsection{Fixed Budget}

For the problem with budget constraint (Problem (\ref{eq:opt_prob_budget})), we plot the variation of fraction of infected nodes at the deadline (the reward function in (\ref{eq:cost_funtion_budget})) and the percentage improvement over the static strategy with respect to the value of the budget in Figs. \ref{fig:J_vs_budget_between_all_nw} and \ref{fig:J_vs_budget} respectively. Due to the fixed budget, the heuristic controls can be uniquely computed (and optimization is not required as was the case earlier). Seeds are uniformly selected from the population for these plots. More budget means abundant resource and less budget corresponds to scarce resource. Hence, we can interpret the behavior of the percentage improvements in Fig. \ref{fig:J_vs_budget} as in Sec. \ref{sec:results_J_vs_b}.

\section{Future Work}

We list some future research directions based on our work:
\begin{enumerate}
\item Developing heuristic low complexity strategies for \emph{dynamic} resource allocation over the period of the campaign for huge social networks that scale to millions of nodes: Some of the insights derived from this work about resource allocation over time and node groups for different system parameters may prove useful in this regard. Also, nodes may be methodically assigned to groups based on the output of some suitably formulated optimization problem to further increase the reward function.

\item In this paper, an individual changes to the `infected' state based only on the states of her neighbors. There are models where an agent's strategy depends not only on the proportion of adopters but also on the collective social position of those adopters \cite{jiang2009concurrent} (\emph{e.g.}, a dominant agent has more effect on the system). Further, cost of influencing different agents in the network may vary \cite{wangpprank}. Devising strategies for optimal information diffusion in such a setting forms an interesting future research direction.

\item This paper takes a structure-oriented view \cite{jiang2014understanding} of behavior diffusion in social networks---social actors are uniform non-strategic nodes. Several authors have modeled nodes of the social network as strategic agents (see for \emph{e.g.} \cite{wang2015studying, jiang2015diffusion}). Devising optimal dynamic information diffusion strategies from the campaigner's perspective in the presence of selfish strategic agents may be of interest.
\end{enumerate}

\section{Conclusion}

We model the spread of information as a susceptible-infected epidemic process in a social network with known adjacency matrix. We use the theory of optimal control to devise strategies for jointly identifying good seeds and allocating campaign resources over time, to groups of nodes, such that a net reward function is maximized. The net reward function is a linear combination of the reward due to spread of a message in the network and the aggregate cost of applying controls. Groups are formed by aggregating nodes with similar values of a centrality measure. The centrality measures used are pagerank, degree, closeness and betweenness. We compare the performance of these centrality measures in information dissemination within the joint optimization--optimal control framework. We also study a problem variant which has a fixed budget constraint.

For the special case when each individual is influenced by a separate control, we show analytically that the controls are non-increasing functions of time. Furthermore, for quadratic costs, the controls are convex functions of time.

Previous benchmark studies on influence maximization considered only seed selection. The system evolved in an uncontrolled manner after seed selection. In contrast, our formulation allows the campaigner to allocate the resource throughout the campaign horizon in addition to seed selection---a situation more general and more practical. Other studies used optimal control for devising strategies for preventing biological epidemics and computer viruses, and in some cases maximizing information dissemination. However, they assumed homogeneous mixing of population which is not valid for information dissemination in social networks. In contrast, we study optimal control of a network with a given adjacency matrix. Our work is also different from the ones that focused on analyzing a given spreading strategy with no optimization/control involved.

We present results for real social networks---scientific collaboration, Slashdot and Facebook networks. We find the optimal strategies to be very effective compared to the simple heuristic strategies for a wide range of model parameters (quantified by the percentage improvements in rewards that the optimal strategies achieve over the heuristics). Our results show that seed and resource allocation based on the simple and local degree centrality performs well compared to other more complicated measures. When the resource is scarce, it is best to target groups with central nodes owing to their spreading strength. On the other hand, if the resource is abundant, groups with non-central nodes are targeted because such nodes are disadvantaged in receiving the message through epidemic spreading. Abundance of resource makes it possible to reach them.

The seed and dynamic resource allocation framework and the solution techniques presented in this paper are general, and can be used to study other ordinary differential equation based information dissemination or biological epidemic models on networks with given adjacency matrix. Examples of other processes running on networks whose dynamic control can be studied using our framework are susceptible-infected-susceptible/recovered (SIS/R), SIRS, Maki-Thompson rumor model etc. In addition, other centrality measures and strategies to group individuals for the purpose of applying controls can be incorporated into the framework. For example, social networks are known to have community structure, each community may form a group influenced by a separate control.

%\small
\footnotesize
\bibliographystyle{IEEEtran}
\bibliography{bibliography_database}

\end{document}